%% file: tcs-generation.tex
\documentclass[a4paper,10pt]{elsarticle}

\input{bbdefs.tex}

\begin{document}

\begin{frontmatter}

\title{\sc Generating Candidate Busy Beaver Machines \\ (Or How to Build the Zany Zoo) }
\author{James Harland \\ 
School of Science \\ 
RMIT University \\
GPO Box 2476 \\ 
Melbourne, 3001 \\ 
Australia \\ 
{\em james.harland@rmit.edu.au} }

\begin{abstract}
The busy beaver problem is a well-known example of a non-computable function. In order to determine a particular value of this function, it is necessary to generate and classify a large number of Turing machines. Previous work on this problem has described the processes used for the generation and classification of these machines, but unfortunately has generally not provided details of the machines considered. While there is no reason to doubt the veracity of the results known so far, it is difficult to accept such results as scientifically proven without being able to inspect the appropriate evidence. In addition, a list of machines and their classifications can be used for other results, such as variations on the busy beaver problem and related problems such as the placid platypus problem. In this paper we investigate how to generate classes of machines to be considered for the busy beaver problem. We discuss the relationship between quadruple and quintuple variants of Turing machines, and show that the latter are more general than the former. We give some formal results to justify our strategy for minimising the number of machines generated, and define a process reflecting this strategy for generating machines. We describe our implementation, and the results of generating various classes of machines with up to 5 states or up to 5 symbols, all of which (together with our code) are available on the author's website. 
\end{abstract}


\end{frontmatter}

\section{Introduction}

Tibor Rado introduced the \textit{busy beaver problem} in 1962 as an example of a non-computable function \cite{Rado62}. It has a surprisingly simple definition: to find the largest number of non-blank characters that is output by a terminating Turing machine, commencing on the blank tape. The machine must be of no more than a given size, so that one can define a function mapping the size of the machine to the maximum size of the output produced. The size of the output can be surprisingly large; for example, there is a machine with six states that terminates on the blank input after $10^{36,534}$ steps and prints out $10^{18,267}$ 1's \cite{Marxen,Harland16}. 

In principle, this problem could be solved by cleverly composing a particular machine of a given size and proving that it is maximal. In practice, solving the problem means analysing all machines of a given size and determining the maximal one (or ones). Ideally the evidence for the maximality of such machines would also be available, in order to allow the result to be checked or reproduced. However, the known results for the busy beaver problem generally do not provide sufficient data or evidence for this to be done. For example, the work of Lafitte and Papazian \cite{LP07} is probably the most comprehensive analysis of the busy beaver problem to date. This includes enumeration of all 3-state 2-symbol and 2-state 3-symbol machines, and some less detailed analyses of larger cases, such as those for 4-state 2-symbol, 2-state 4-symbol and 3-state 3-symbol machines. Unfortunately though, there appears to be no code available, nor data files, and whilst they provide a description of the processes used, it is difficult to accept the results provided as complete proof of the properties claimed. 

A similar property applies to older results for the busy beaver as well. Lin and Rao \cite{LR65} provide the earliest systematic analyses of the 3-state 2-symbol case, which involved using a program which was able to analyse all but 40 machines, which were then analysed by hand. They provide a description of their method and a specification of the 40 machines, but the details provided are not sufficient to reproduce exactly what was done, and the code used does not seem to be available. Brady \cite{Brady83} and Machlin and Stout \cite{MS90} provide analyses of the 4-state 2-symbol machines, which also use programs to reduce the unknown cases to 218 and 210 respectively. Both times, these remaining machines were determined not to terminate due to a human analysis, but for which there is no direct evidence, or even a specification of exactly which machines were involved, and the code used in these analyses also does not seem to be available. 

We have no reason to doubt the correctness of these results, or that they were obtained in good faith, using the best techniques possible at the time. However, we believe that it is not scientifically proper to accept these results as proven, given that there is neither a mathematical proof of their correctness nor sufficient empirical evidence that would allow their claims to be inspected, assessed and checked. It should also be said that the computing resources available now dwarf those of even the recent past, and that many tasks which are now routine were considered unthinkable even ten years ago, let alone fifty. By the same token, those same computing resources which are now available in this era of cloud computing mean that any claims about busy beaver results should be required to provide a substantially higher level of ``computational evidence'' than those described above. 
\textbf{This evidence should include not only the programs used for the search, but also the list of machines generated as well as the evidence used to draw conclusions about their status.} In other words, in an era when the Flyspeck project has provided a (very large) computer-assisted proof of a long-standing conjecture of Kepler's \cite{Hales05,Flyspeck}, it seems particularly important for the evidence for all claims about the busy beaver and related problems to be made on the basis of verifiable and reproducable evidence, including all relevant code and data. A similar conclusion has been reached by de Mol, who has argued that computer-assisted proofs such as these should include a \textit{description of the computational process}, its \textit{output} and the \textit{code used} \cite{deMol}.

As we have argued previously \cite{Harland16}, this means that there is a need to re-examine our knowledge of the busy beaver problem, and to provide the appropriate level of evidence for the results derived (as well as to determine new results in a similar manner, of course). This will involve providing sufficient detail 
(including code and data) so that the results can be verified, or reproduced if desired, by an interested researcher. 
In this paper we address the problem of how to generate the classes of machines needed in order to determine the busy beaver function, corresponding to steps 1 and 2 of the framework in \cite{Harland16}. In principle, we can do this independently of the execution of machines. Generating the machines independently from the execution and analysis of them not only seems intuitively natural but also allows for different analyses of the same class of machines, possibly by different researchers. This ability to reproduce and hence confirm or refute empirical results seems fundamental to the busy beaver problem. 

\begin{figure}[tb]
\centering
\includegraphics[width=0.7\textwidth]{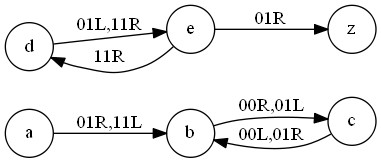}
\caption{Disconnected machine}
\label{fig:example1}
\end{figure}

In practice, the separation between these two processes is not so neat and simple. 
Firstly, the number of machines of a given size grows very quickly; in fact, there are $O(n^n)$ machines of size $n$. Hence any saving that can be made by reducing the number of machines to be processed is vitally important. 
Secondly, the process for generating machines needs to be carefully considered and based on sound principles, rather than simply generating all possible machines of a given size. 
For example, consider the 5-state 2-symbol machine in Figure~\ref{fig:example1}. As there is no connection between states $a,b,c$ and $d,e$, this is effectively a machine of size 3, rather than size 5. Such redundancies clearly need to be avoided if possible. 
Thirdly, in order to determine the busy beaver function, we need only analyse the execution of these machines on a single input (i.e.\ when the tape is entirely blank). This means that we can potentially exploit this property to reduce the number of machines that need to be considered. 

The classic way to address these issues is to intermingle the generation of the machines with their execution by use of a technique known as \textit{tree normal form (tnf)}\cite{Rado62,LR65}. The idea is to execute a partially defined machine on the blank input (with some suitable initial set of transitions) until an unallocated transition is found, at which point an appropriate additional transition is allocated, and various checks are performed on the resulting machine. If the resulting machine fails one of these checks, an alternative allocation of an additional transition is sought. Otherwise, execution then proceeds with the extended machine. This process continues until the machine is sufficiently defined and a halting transition is added to the machine, at which point the machine generation is complete. In order to find all such machines, it is straightforward to perform a backtracking search over all choices of machine that can be made. This process of lazy generation of machines significantly reduces the number of machines that need to be considered (see Section~\ref{sec:results}). 

The \textit{tnf} approach makes it attractive to use a single process to generate and analyse machines rather than separating the two as discussed above, and this is the method that has generally been used in the past \cite{LR65,Brady83,MS90,LP07}. However, this has the disadvantage of not allowing multiple analyses of the same data. This is important not just for verification of completed results, but also for developing the methods to derive such results. In other words, it is often useful to be able to evaluate the effectiveness of a given analysis method by testing it on a large number of examples, identifying those machines for which the analysis is problematic or incomplete, and using these examples to refine and improve the analysis method. A ``one-shot'' approach will suffice if a complete analysis procedure, which is known in advance to be capable of analysing all machines of interest, is available. This is generally unrealistic in practice, and so it seems necessary to have an incremental development cycle involving testing a partially developed technique on a class of machines, and using the results to further improve the effectiveness of the technique. This makes it appropriate to separate the generation and analysis of the machines of interest as best we can. 

Our intention is to provide a data set consisting of all machines which need to be considered for the busy beaver problem of a given size. Clearly it is important to minimise the size of this set where possible. However, it is even more important to ensure that all relevant machines are present, or equivalently that any omissions from this set are certainly irrelevant. This means that we cannot generally reduce the number of irrelevant machines retained for analysis to 0, as we will only eliminate a machine from consideration if we are sure it is irrelevant. For this reason we investigate formal results which establish the correctness of constraints on the type of machines that need to be generated in order to solve the busy beaver and related problems. We define what is meant by being relevant to the busy beaver problem, and give results which establish the soundness of eliminating certain classes of machines from consideration. We provide a procedure for generating machines which respects these constraints, and report our results obtained from an implementation of it. The code used and the results obtained are available from the author's website.\footnote{\url{www.cs.rmit.edu.au/~jah/busybeaver}}

This paper is organised as follows. In Section~\ref{sec:related} we discuss related work, and in Section~\ref{sec:definitions} we discuss some preliminaries. 
In Section~\ref{sec:qq} we investigate the relationship between quadruple and quintuple variants of Turing machines, and in
Section~\ref{sec:normal} we show how to obtain a normal form for the machines to be generated.
In Section~\ref{sec:monotonicity} we discuss some issues that arise from the process of generating machines, and in Section~\ref{sec:generation} we specify a procedure for generating relevant machines.  
In Section~\ref{sec:results} we present the results of our generation of various classes of machines, and
in Section~\ref{sec:conclusion} we present our conclusions and discuss possibilities for further work. 

\section{Related Work}\label{sec:related} 

The busy beaver function is defined as the maximum number of non-blank characters that is printed by a terminating $n$-state Turing machine on the blank input. 
This function is often denoted as $\Sigma(n)$; in this paper we will use the more intuitive notation of $bb(n)$ (as in \cite{Harland16}). 
The number of non-blank characters printed by a terminating machine is known as its {\em productivity} \cite{BBJ}. The number of steps taken by the machine to terminate on the blank input is known as its \textit{activity} \cite{Harland16}. A non-terminating machine has activity $\infty$. 
In the literature, the maximum activity for machines of size $n$ is often denoted as ${\cal S}(n)$; in this paper, we denote this function as $\mathit{ff}(n)$.\footnote{We call this function the {\em frantic frog.}} As we will be considering machines with both a varying number of states and a varying number of symbols, we will generally use the notation $bb(n,m)$ and $\mathit{ff}(n,m)$ for a machine with $n$ states and $m$ symbols. 

Rado's motivation for introducing the busy beaver function was to have a specific example of a non-computable function. That such functions exist had been known for some time \cite{Church36,Turing36}, but Rado was interested to find a particularly simple definition of a specific function in this class. He was able to show that the busy beaver function is non-computable by showing that it grows faster than any computable function. Not long afterwards, Lin and Rado \cite{LR65} were the first to produce some specific values for the busy beaver itself, for the cases in which the machine had 1, 2 or 3 states (and exactly 2 symbols). Later, Brady \cite{Brady83} and independently Machlin and Stout \cite{MS90} confirmed these results, and extended them to include the case for 4 states. This also included a more sophisticated analysis of non-terminating machines. Lafitte and Papazian \cite{LP07} provided the most comprehensive examination of these classes of machines to date, which also included the only known systematic analysis thus far of machines with 3 states and 3 symbols. They have also provided some intuitive explanation of why the 2-state 4-symbol class is the potentially the most difficult to analyse of these classes, as well as some empirical results about the apparent square law relationship between activity and productivity for the most productive machines. 

The machines with 5 or 6 states have also received some attention, although generally in less detail. Marxen and Buntrock  \cite{HM90} provided the first comprehensive analysis of the 5-state case, and to a lesser extend, of the 6-state one as well. This work sparked off an informal competition to find machines of very large productivity, which continues to this day, and is very well documented in great detail on Marxen's website \cite{Marxen}. The machines of greatest activity and productivity given in \cite{HM90} have been superceded over the years by a series of increasingly productive (and active) machines, found by various contributors in unpublished work. There is also some analysis of this class of machines in the work of Michel \cite{Michel93,Michel}. 
There seems to be an informal consensus that the maximum productivity for 5-state 2-symbol machines is 4,098 (although there are two machines with this productivity, with activity 47,176,870 and 11,798,826), but for other classes there are machines with productivity and activity much larger than this. 
Kellett \cite{Kellett}, building on the earlier work of Ross \cite{Ross}, has provided an analysis of the 5-state machines, using so-called ``quadruple'' machines, (i.e.\ each transition is a quadruple rather than a quintuple, reflecting the fact that each transition can either change the symbol on the tape or move the tape head, but not both). As with seemingly all previous studies in this area, the classification is not fully automated; in particular, it is reported that the 98 unclassified machines were analysed by hand, rather than by an automated procedure. Kellett also reports some results for 6-state machines, but without any analysis of the non-terminating machines. 

It is well-known that the quadruple and quintuple variants of Turing machines are equivalent from a computability perspective (i.e.\ for any machine in one class there is an equivalent machine in the other), it is not obvious that this equivalence is maintained when the number of states is restricted; in particular, it is not obvious that for any \textit{5-state quintuple} machine that there is an equivalent \textit{5-state quadruple} machine. This means that it is appropriate, if perhaps a little conservative, to consider the quadruple and quintuple cases as separate problems. There are some quadruple machines of notable size on Marxen's web page, but it is also worth noting that the machines of largest known productivity are all quintuple machines. We discuss this issue in more detail in Section~\ref{sec:qq}, where we will show that the quintuple problem subsumes the quadruple one. 

\section{Preliminaries and Definitions}\label{sec:definitions} 

We use the following definition of a Turing machine, which is essentially that of Sudkamp \cite{Sudkamp} with some minor variations as discussed in \cite{Harland16}. Here we assume the tape alphabet and the input alphabet are the same. 

\begin{definition} \label{def:tm}
A Turing machine is a quadruple $(Q \cup \{z\}, \Gamma, \delta, a)$ where

\begin{itemize}
\item $z$ is a distinguished state called a {\em halting state}
\item $\Gamma$ is the tape alphabet
\item $\delta$ is a partial function from $Q \times \Gamma$ to $Q \cup \{z\} \times \Gamma \times \{l,r\}$ called the {\em transition function}
\item $a \in Q$ is a distinguished state called the {\em start state}
\end{itemize} 
\end{definition}

Note that due to the way that $\delta$ is defined, this is the so-called \textit{quintuple transition variation} of
Turing machines, in that a transition must specify for a given input state and input character, a new state, an output character and a
direction for the tape in which to move. Hence a transition can be specified by a quintuple of the form

\begin{center}
$(State, Input, Output, Direction, NewState)$
\end{center}
where $State \in Q$, $NewState \in Q \cup \{z\}$, $Input, Output \in \Gamma$ and $Direction \in \{l,r\}$. 

We call a transition a \textit{halting transition} if $NewState = z$; otherwise it is a \textit{standard transition}. 

The quadruple machines mentioned above allow only one of the $Output$ or $Direction$ to be specified in a transition, i.e.\ that each transition must either write a new character on the tape or 
move, and not both; for such machines, clearly only a tuple of 4 elements is required, which is why we refer to this case as the \textit{quadruple transition variation}. 

Given some notational convention for identifying the start state and halting state, a Turing machine can be characterised by the tuples which make up the definition of $\delta$. In this paper the start state is denoted $a$ and the halt state is denoted $z$. We also denote the blank symbol as $0$. 
We will also use $\_$ to indicate a variable whose values may be arbitrary (a la Prolog \cite{swi}), so that for example a tuple in which \textit{Input} is 1, \textit{Output} is 0 and all other values are arbitrary would be denoted $(\_, 1, 0, \_, \_)$. 

We denote by an $n$-state Turing machine one in which $|Q| = n$. In other words, an $n$-state Turing machine has $n$ standard states and a halting state.  
Note also that there are no transitions from state $z$, and that as $\delta$ is a partial function, there is at most one transition for a given
pair of a state and a character in the tape alphabet. This means that our Turing machines are all \textit{deterministic}; there is never a case when a machine has to ``choose'' between two possible transitions for a given state and input symbol. Note that it is possible for there to be no transition for a given state and symbol combination. If such a combination is encountered during execution, the machine halts.  We generally prefer to have explicit halting transitions rather than have the machine halt in this way. 

A \textit{configuration} of a Turing machine is the current state of execution, containing the current tape contents, the current
machine state and the position of the tape head \cite{Sudkamp}. We will use $111\{b\}011$ to denote a configuration in which the Turing machine is
in state $b$ with the string 111011 on the tape and the tape head pointing at the 0. In this paper, we do not formally define how an initial configuration and a (deterministic) Turing machine can be used to specify a possibly infinite sequence of configurations, as this is well-known and is not of central importance. We will refer to this sequence of configurations as a \textit{computation}, or the \textit{execution} of the machine. 

Machine states are labelled $a,b,c,d,e \ldots$ where $a$ is the initial state of the machine. The halting state is labelled $z$. Symbols are labelled $0,1,2,3,\ldots $ where $0$ is the blank symbol. 

We will find the following notion of \textit{dimension} useful.  

\begin{definition}
Let $M$ be a Turing machine with $n$ states and $m$ symbols (where $n,m \geq 2$). Then we say $M$ has \textit{dimension} $n \times m$. 
\end{definition}

The busy beaver problem is based on terminating machines of a given size. Hence a vital aspect of generating the appropriate machines to consider is the how the halting transitions are introduced to the machine, and the point in the process when this is done. This is the motivation for the definition below. 

\begin{definition}

We say a Turing machine $M$ is

\begin{itemize}
	\item \textbf{k-halting} if there are $k$ transitions of the form $(\_,\_,\_,\_,z)$ in $M$. 
	\item \textbf{exhaustive} if $\delta$ is a \textbf{total} function from $Q \times \Gamma$ to $Q \cup \{z\}
  \times \Gamma \times \{l,r\}$, i.e.\ that \\
  $\forall q \in Q \; \forall \gamma \in \Gamma, \; \exists q' \in Q \cup \{z\}, \gamma' \in \Gamma$ and $D \in \{l,r\}$ such that $\delta(q,\gamma) = \langle q',\gamma',D \rangle$. 
	\item \textbf{n-state full}  if $|Q| = n$.
	\item \textbf{m-symbol full} if $|\Gamma| = m$.
\end{itemize}
  
\end{definition}

Note that machines which are exhaustive but not 1-halting are either guaranteed not to terminate (as there is no transition into the halting state and every combination of state and input symbol has a transition defined for it), or have multiple halting transitions, of which at most only one can ever be used in a given computation, making the other halting transitions spurious. It is interesting to note that the Wolfram prize for finding minimal universal Turing machines \cite{wolframprize} used machines which are exhaustive and 0-halting, and hence these machines can never terminate\footnote{Termination is simulated by repeatedly reproducing the final simulated configuration.}. In our case, we will seek to generate $k$-halting machines for some $k > 0$, and preferably with $k = 1$. 

In this paper we will assume that an $n$-state $m$-symbol Turing machine has $n \geq 2$ and $m \geq 2$. 
When discussing exhaustive machines, we will sometimes find it useful to explicitly refer to $n$ and $m$, so we call a Turing machine $n$-$m$-exhaustive when it is exhaustive and we have $|Q| = n$ and $|\Gamma| = m$.

Note that as we will be generating machines incrementally, we may find that we have a machine containing, say, 3 states and 2 symbols during the process of generating a 5-state 2-symbol machine. In this case, the generated machine is 3-state full (and 2-symbol full) but not 4-state full or 5-state full, and as a 3-state 2-symbol machine can have at most 6 transitions, it is not 5-2-exhaustive. 

As we are interested in finding the maximum number of non-blank characters that can be printed by a terminating machine, we will use the concept of \textit{maximising} machines, defined below. 

\begin{definition}
Let $M$ be a Turing machine. 

\begin{itemize}
	\item $M$ is a \textbf{maximising} machine if for every tuple of the form $(\_,\_,O,\_,z)$ in $M$, $O \not = 0$. 
	\item $M$ is a \textbf{minimising} machine if for every tuple of the form $(\_,\_,O,\_,z)$ in $M$, $O$ is 0. 
\end{itemize}
\end{definition}

The final transition that is executed in a terminating computation in a maximising machine, i.e.\ the halting transition, can be guaranteed not to reduce the number of non-blank symbols on the tape. This seems an entirely desirable property in a search for busy beavers. Note that when generating machines, we can satisfy the requirement for the machine to be maximising by ensuring that $O = 1$ in any halting transition. This is a convenient choice, as we can be sure that $1$ will occur in every machine we consider, but not necessarily any other non-blank symbol. 
In some circumstances it can be useful to consider minimising machines, but this is comparatively rare. Note that if there is more than one halting transition, it is possible for the machine to be neither maximising nor minimising. 

When searching for busy beaver machines, it seems natural to focus on machines which are \textit{1-halting, exhaustive} and \textit{maximising}. As there is only one input on which the machine is evaluated (the blank input), there is only ever one halting transition that can be executed. This means that a $k$-halting machine where $k > 1$ is potentially wasteful. in that such transitions may be better used to define a larger and more complex machine. For similar reasons, it seems natural to focus on exhaustive machines, as non-exhaustive ones are also potentially wasteful. As discussed above, it also seems natural to insist on maximising machines. It is worth noting that all known ``monster'' busy beaver machines are 1-halting, exhaustive and maximising \cite{Marxen,Harland16}, although to the best of the author's knowledge, this is merely an empirical observation rather than a proven result. We will use the term \textit{dreadful dragons} to refer to Turing machines of very high productivity and activity. We will also refer to the $n$th entry in the table of 100 dreadful dragons in \cite{Harland16} as Dragon $n$. 

Note that in a 1-halting, exhaustive and maximising machine, we may assume that the (unique) halting transition is always of the form $(S,I,1,D,z)$, and we may arbitrarily assign $D$ as $r$. This means that the halting transition can be completely specified once $S$ and $I$ are known. Note that all 1-halting, exhaustive and maximising machines will have exactly $n \times m$ transitions, one for each combination of state and tape symbol, and with $(n \times m) - 1$ standard transitions, and one halting transition. 

It should be noted that if we consider all possible machines, rather than those generated by the \textit{tnf} process, there are as many $n$-state $m$-symbol machines as there are $m$-state $n$-symbol machines. When determining values for each of \textit{State, Input, Output, Direction} and \textit{NewState} in the tuple \textit{(State, Input, Output, Direction, NewState)}, for an $n$-state $m$-symbol machine there are $n$ choices for each of \textit{State} and \textit{NewState}, and $m$ choices for each of \textit{Input} and \textit{Output} whereas for an 
$m$-state $n$-symbol machine there are $m$ choices for each of \textit{State} and \textit{NewState}, and $n$ choices for each of \textit{Input} and \textit{Output}. This means that if we generate the machines naively, we will find that there are the same number of machines with say 5 states and 2 symbols as there are with 2 states and 5 symbols (although the use of the \textit{tnf} process disrupts this neat symmetry). This is why we will speak of generating machines of a given dimension, rather than a given number of states. 

It should also be noted that we have a rather unusual property for Turing machines, which is that we are only interested in their execution on the blank input. This means that we can make some significant reductions on the machines that need to be considered (see Section~\ref{sec:normal}). 

Not all machines will be relevant to the busy beaver problem, and in fact there are some terminating machines which will be of little interest. This leads us to the definition below. 

\begin{definition}
Let $M$ be a Turing machine. We say $M$ has {\bf activity} $k$ if the execution of $M$ on the blank input terminates in $k$ steps. If the execution of $M$ on the blank input does not terminate, we say the activity of $M$ is $\infty$. The {\bf productivity} of a machine $M$ whose activity is finite is the number of non-blank characters on the tape after $M$ has halted. If the activity of $M$ is $\infty$, then the productivity of $M$ is not defined. 

We denote the activity and productivity of a Turing machine $M$ as {\em activity(M)} and {\em productivity(M)} respectively. 

We say two machines $M_1$ and $M_2$ are {\bf activity equivalent} if $activity(M_1)$ and $activity(M_2)$ are both $\infty$, or if $activity(M_1)$ and $activity(M_2)$ are both finite and $activity(M_1) = activity(M_2)$. 

We say two machines $M_1$ and $M_2$ are {\bf productivity equivalent} if $productivity(M_1)$ and $productivity(M_2)$ are both undefined, or $productivity(M_1)$ and $productivity(M_2)$ are both defined and $productivity(M_1) = productivity(M_2)$. 
\end{definition}

Note that the productivity of a Turing machine is only defined for machines which terminate on the blank input. This is slightly different to the productivity defined in \cite{BBJ} (page 42), in which the productivity is defined for contiguous non-blanks, and the productivity of a non-terminating machine is $0$. We believe it is more intuitive for the productivity only to be defined for machines which terminate on the blank input, and that it is appropriate to define productivity to include discontiguous non-blanks (i.e.\ that the non-blank symbols can be on any pattern on the tape). This is consistent with the machines of very high productivity reported on Marxen's web page \cite{Marxen}. 

We are now in a position to define restrictions on the machines generated which will eliminate some uninteresting cases. 

\begin{definition}\label{def:blank}
We say a Turing machine $M$ {\bf satisfies the blank tape condition} iff during the computation of $M$ on the blank input, the only configuration in which the tape is blank is the initial configuration. 
\end{definition}

\begin{definition}\label{def:relevant}
We say an $n$-state Turing machine $M$ where $n \geq 2$ is {\bf irrelevant to the busy beaver problem} if at least one of the following conditions is satisfied.
\begin{itemize}
	\item $activity(M)$ is $\infty$
	\item $activity(M) \leq n$
	\item $productivity(M) = 0$
	\item $M$ does not satisfy the blank tape condition
\end{itemize}
Otherwise it is {\bf relevant to the busy beaver problem.}
\end{definition}

This means that in order to be relevant to the busy beaver problem, a Turing machine $M$ must be such that $activity(M)$ is finite and $> n$, $productivity(M) > 0$ and $M$ must satisfy the blank tape condition.
In order to solve the busy beaver problem, clearly machines of productivity 0 are of little interest. Similarly machines of activity 1 or 2, which terminate after 1 or 2 steps of computation, are also of little interest. It is also simple to construct an $n$-state machine of activity $n$ and productivity $n$ as follows. We label the states as $a_1 = a, a_2, \ldots, a_n$ for convenience. 

\begin{center}
\begin{tabular}{ccccc}
\textbf{State} & \textbf{Input} & \textbf{Output} & \textbf{Direction} & \textbf{New State} \\
$a_1$ & 0 & 1 & r & $a_2$ \\
$a_2$ & 0 & 1 & r & $a_3$ \\
\ldots & & & & \\
$a_n$ & 0 & 1 & r & $z$ \\
\end{tabular}
\end{center}

Note that this is actually a template for many machines; as the machine will only be executed on the blank input, all transitions other than those for which \textbf{Input} is 0 will not be executed. Any machine satisfying this template terminates in $n$ steps with $n$ 1's on the tape, as follows. The notation $C_1 \Rightarrow C_2$ is used to indicate that configuration $C_2$ is the result of one execution step in configuration $C_1$. 

\begin{center}
$\{a_1\}0 \Rightarrow 1\{a_2\}0 \Rightarrow 11\{a_3\}0 \Rightarrow \ldots \Rightarrow 1^{n-1} \{a_n\}0 \Rightarrow 1^n 0 \{z\}$
\end{center}

\begin{figure}
\centering
\includegraphics[width=0.8\textwidth]{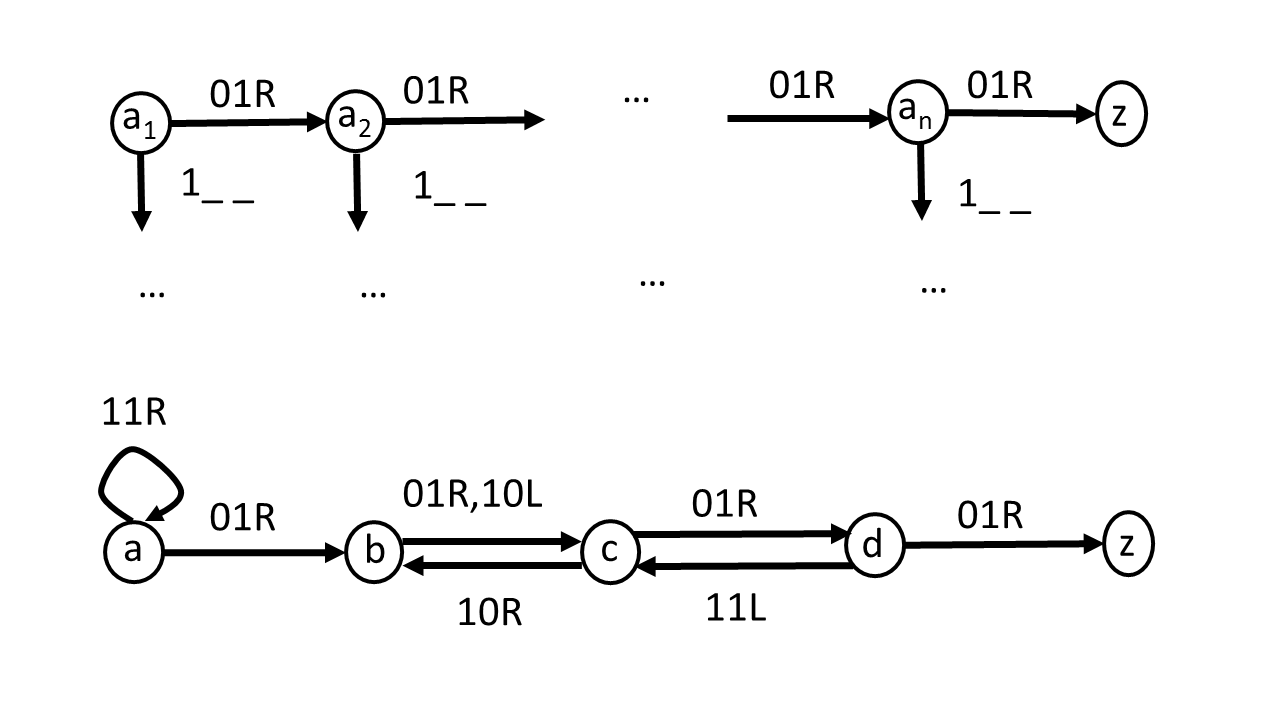}
\caption{Gutless goanna template and example machine}
\label{fig:goanna}
\end{figure}

A diagrammatic representation of this machine template and an example of a machine which conforms to this template are given in Figure~\ref{fig:goanna}. The similarity between this representation and a goanna\footnote{A goanna is an Australian lizard.} is the reason for the name \textit{gutless goanna} for this kind of machine.

This means that we should concentrate our search on $n$-state machines of activity $> n$. It may also seem natural to insist that relevant machines have productivity $> n$ rather than just $> 0$. However, it is generally much harder to guarantee a minimum level of productivity than it is to guarantee a minimum level of activity. 

The most unusual aspect of Definition~\ref{def:relevant} is the blank tape condition. The reason that this condition is included is 
to minimise the number of machines that need to be considered. For example, consider a machine in which the tape remains blank until the third step of execution (i.e.\ the first two such steps do not alter the tape, which is initially blank). As we shall see in Section~\ref{sec:normal}, it is straightforward to transform this machine into one which does not have this property, and yet has the same productivity as the original. The new machine will, though, have a lower activity than the original, but the new machine does not ``waste'' the first two steps of execution. This may be thought of as only considering the part of the execution of the machine that involves a non-blank tape (apart from at the very beginning) as appropriate for consideration for the busy beaver problem. Taking this approach a little further, we can also perform a similar transformation on a machine whose execution returns the configuration to one containing a blank tape at some point after the initial configuration. Again, the transformed machine will have the same productivity but a lower activity, and may be thought of as only considering the part of the computation that involves a non-blank tape. Details of these transformations and proofs that they preserve productivity are given in Section~\ref{sec:normal}. 

\begin{figure}
\centering
\includegraphics[width=0.7\textwidth]{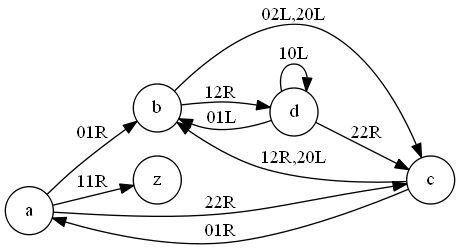}
\caption{Dragon 91}
\label{fig:m91raw}
\end{figure}

The important consequence of requiring the blank tape condition is that it will reduce the number of machines that are considered. This introduces some minor variation into the precise definition of some aspects of the problem. As mentioned above, the transformations do not change productivity, and so the value of $bb$ will not be altered. However, our approach may result in some slightly lower values of $\mathit{ff}$, and in general some slightly lower values for activity than may be found elsewhere, such as some of the activity values for the known dreadful dragons  reported by Marxen \cite{Marxen}. For example, consider the machine in Figure~\ref{fig:m91raw}, which is Dragon~91. This machine returns the tape to blank after 5 steps, and then proceeds for around $10^{14,000}$ steps before halting. The initial steps of the execution of this machine are below. 

\begin{center}
$\{a\}0 \Rightarrow 1\{b\}0 \Rightarrow \{c\}12 \Rightarrow 2\{b\}2 \Rightarrow \{c\}20 \Rightarrow \{b\}0 \Rightarrow \{c\}02 \Rightarrow 1\{a\}2 \Rightarrow \ldots$
\end{center}

This computation can be considered as ``restarting'' in configuration $\{b\}0$ after 5 steps of execution from $\{a\}0$. 
Clearly transforming the machine to avoid this will reduce the number of steps executed by 5, but does not change the productivity of the machine (which is around $10^{7,000}$). 

While it is certainly true that maximising computational properties (and hence maximising the value of $\mathit{ff}$ by whatever means) is very much in the spirit of the busy beaver problem, it is also a fundamental aspect of the problem to define precisely the class of Turing machines of interest, and hence precisely define the problem. As with many issues, there is no obvious ``right'' way to choose such a definition;  in our case, we prefer one which when compared to the machines on Marxen's website 
both preserves productivity (and hence $bb$ values) and reduces the number of machines to be considered over one that preserves both productivity and activity at the cost of an increased number of machines to be analysed. 

\section{Quintuple versus Quadruple}
\label{sec:qq}

A further issue that arises in discussion about the precise definition of Turing machine that is to be used is whether to the quintuple or quadruple variant. As noted above, some previous analyses have used the quadruple variant of Turing machines. An attraction of this approach is that there are generally less machines to consider (as there are only 4 elements to consider for each transition rather than 5). However, the precise relationship between the two variants with respect to the busy beaver problem is not obvious, especially as all the known dreadful dragons are quintuple machines. In this section we explore the relationship between the two types of Turing machine, and show that searching quintuple machines will find a superset of those found by searching quadruple machines. 

Ross \cite{Ross} and Kellett \cite{Kellett} have given an analysis of 5-state 2-symbols machines based on the quadruple variant of Turing machines. 
It is well-known that the quadruple and quintuple variants of Turing machines are equivalent, i.e.\ that for every quadruple machine there is an equivalent quintuple machine, and vice-versa \cite{Ross}. However, this result is not strong enough for our purposes, as it is not clear that this equivalence is maintained if we impose the extra restriction of requiring the transformation to preserve the number of states and symbols in the Turing machine. 
In particular, it is not clear that for a given $n$-state $m$-symbol quintuple machine that there is an equivalent \textit{$n$-state $m$-symbol} quadruple machine. It is also worth noting that whilst there are some quadruple machines of notable size \cite{Marxen}, the largest known machines in all classes are all defined as quintuple machines.

With a little care, it is not overly difficult to show that for an $n$-state $m$-symbol quadruple machine that there is an equivalent $n$-state $m$-symbol quintuple machine. We do this below. 

\begin{definition} \label{def:tm4}
A {\bf quadruple} Turing machine is a quadruple $(Q \cup \{z\}, \Gamma, \delta, a)$ where

\begin{itemize}
\item $z$ is a distinguished state called a {\em halting state}
\item $\Gamma$ is the tape alphabet
\item $\delta$ is a partial function from $Q \times \Gamma$ to $Q \cup \{z\} \times \Gamma \cup\{l,r\}$ called the {\em transition function}
\item $a \in Q$ is a distinguished state called the {\em start state}
\end{itemize} 

We say a transition $(S,I,OD,NS) \in M$ is a {\bf movement transition} if $OD \in \{l,r\}$. Otherwise we say it is an {\bf output transition}.
\end{definition}

Note that the only difference between this definition and Definition~\ref{def:tm} is that $\delta$ is now a function that results in a state and either a direction or an output, but not both. 

An interesting point to note is that for an output transition $(S, I, O, NS)$, the next configuration will have the machine in state $NS$ with input $O$. This means that if there are two output transitions in $M$ of the form $(S, I, O_1, S_1)$ and $(S_1, O_1, O_2, S_2)$, then we obtain an equivalent machine by replacing the first transition with $(S, I, O_2, S_2)$. Clearly the second machine will perform less steps than the first, but the same set of final configurations will be reached. What follows is a formalisation of this idea. 

We note the following about sequences of output transitions.

\begin{definition} \label{def:tm4cycle}
Let $M$ be a quadruple Turing machine. We say 

\begin{itemize}
	\item $S_0,I_0$ is in an {\bf output cycle in $M$} if there is a sequence of output transitions in $M$ of the form 
$(S_0, I_0, O_1, S_1), (S_1, O_1, O_2, S_2), (S_2, O_2, O_3, S_3), \ldots (S_k, O_k, I_i, S_i)$ for some $0 \leq i \leq k$.

	\item $S_0,I_0$ is in an {\bf output chain in $M$} if there is a sequence of output transitions in $M$ of the form 
	
$(S_0, I_0, O_1, S_1), (S_1, O_1, O_2, S_2), (S_2, O_2, O_3, S_3), \ldots (S_{k-1}, O_{k-1}, O_k, S_k)$  

where there is no transition of the form $(S_k, O_k, \_, \_)$ in $M$, and $\forall 1 \leq i \leq k$, $(S_0, I_0, O_1, S_1), \ldots (S_{i-1}, O_{i-1}, O_i, S_i)$ is not an output cycle.

	\item $S_0, I_0$ is in a {\bf movement chain in $M$} if $(S_0,I_0,D,S_1)$ is a movement transition, or there is a sequence of transitions in $M$ of the form 
	
$(S_0, I_0, O_1, S_1), (S_1, O_1, O_2, S_2), (S_2, O_2, O_3, S_3), \ldots (S_{k-1}, O_{k-1}, O_k, S_k), (S_k, O_k, D, S_{k+1})$ 

where 
$(S_k, O_k, D, S_{k+1})$ is a movement transition, 
all the other transitions are output transitions, and
$\forall 1 \leq i \leq k$, $(S_0, I_0, O_1, S_1), \ldots (S_{i-1}, O_{i-1}, O_i, S_i)$ is not an output cycle.

\end{itemize}
\end{definition}

Note that a movement chain is basically an output chain which is followed by an appropriate movement transition. We may also think of output chains and output cycles as two exclusive and exhaustive classes of sequences of output transitions. 

It is then simple to show the following result. 

\begin{lem}\label{lemma:output}
Let $M$ be a quadruple Turing machine. Then for any transition $(S,I,\_,\_)$ in $M$, $S,I$ is in either an output cycle, an output chain or a movement chain in $M$. 
\end{lem}

\begin{proof}
If $(S,I,OD,NS)$ is a movement transition, then $S,I$ is trivially in a movement chain in $M$. Otherwise, $(S,I,OD,NS)$ is an output transition. We then follow the chain of transitions starting at $NS, OD$ until we either find there is no transition, encounter a movement transition, or encounter an output cycle. This results in either an output chain, a movement chain, or output cycle respectively. \qed
\end{proof}

\begin{definition}\label{definition:normalised}
Let $M$ be a quadruple Turing machine. We say $M$ is {\bf normalised} if for every output transition $(S,I,OD,NS) \in M$ one of the following conditions holds.
\begin{itemize}
	\item $S = NS$ and $I = OD$
	\item There is no transition for $NS,OD$
	\item There is a movement transition in $M$ of the form $(NS,OD,\_,\_)$
\end{itemize}
\end{definition}

It is then straightforward to show the following Proposition.

\begin{prop} \label{prop:quadnormal}
Let $M$ be a quadruple Turing machine. Then there is a normalised machine $M'$ which is productivity equivalent to $M$.
\end{prop}

\begin{proof}
By Lemma~\ref{lemma:output}, for any transition $(S,I,\_,\_)$ in $M$, $S,I$ is either contained in an output cycle, an output chain or a movement chain. We generate $M'$ from $M$ as below.

\begin{itemize}
	\item If $(S,I,OD,NS) \in M$ and $S,I$ is in an output cycle, replace $(S,I,OD,NS)$ with $(S,I,I,S)$. 
	\item If $(S,I,OD,NS) \in M$ and $S,I$ is in an output chain $(S,I,OD,NS), \ldots (S_{k-1}, O_{k-1}, O_k, S_k)$, replace $(S,I,OD,NS)$ with $(S,I,O_k,S_k)$.
	\item If $(S,I,OD,NS) \in M$ and $M$ has a movement chain $(S,I,OD,NS), \ldots (S_{k-1}, O_{k-1}, O_k, S_k), (S_k, O_k, D, S_{k+1})$, replace $(S,I,OD,NS)$ with $(S,I,O_k,S_k)$
\end{itemize}

This generates a normalised machine $M'$ which is productivity equivalent to $M$. 
\qed
\end{proof}

For example, consider the three sets of transitions below.

\begin{center}
$(a,0,1,b), (b,1,0,c), (c,0,1,d), (d,1,0,e), (e,0,0,a)$ \\
$(a,0,1,b), (b,1,0,c), (c,0,1,d), (d,1,0,e)$ \\
$(a,0,1,b), (b,1,0,c), (c,0,1,d), (d,1,0,e), (e,0,r,a)$
\end{center}

The first is an output cycle. The second, as there is no transition for $e,0$, is an output chain. The third is a movement chain. Under the transformation of Proposition~\ref{prop:quadnormal}, these would be replaced by the transitions below.

\begin{center}
$(a,0,0,a), (b,1,1,b), (c,0,0,c), (d,1,1,d), (e,0,0,e)$ \\
$(a,0,0,e), (b,1,0,e), (c,0,0,e), (d,1,0,e)$  \\
$(a,0,0,e), (b,1,0,e), (c,0,0,e), (d,1,0,e), (e,0,r,a)$
\end{center}

Clearly if an output cycle is encountered during execution, then the machine will never terminate. This leads us to the following result. 

\begin{prop}\label{prop:4nocrap}
Let $M$ be a normalised quadruple machine, and let $M'$ be the machine resulting from deleting all transitions of the form $(S,I,I,S)$ from $M$. If $M$ terminates on the blank input, then $M'$ also terminates on the blank input. 
\end{prop}

\begin{proof}
Clearly as $M$ terminates on the blank input, no transition of the form $(S,I,I,S)$ is used in this execution. Hence deleting these transitions from $M$ will have no effect. 
\qed
\end{proof}

We can now show the main result of this section, which is that a normalised quadruple machine can be easily rewritten as a quintuple machine which is productivity equivalent. 

\begin{prop}\label{prop:4to5}
Let $M$ be a normalised quadruple machine containing no transition of the form $(S,I,I,S)$. Then there is a quintuple machine $M'$ that is productivity equivalent to $M$.
\end{prop}

\begin{proof}
As $M$ is normalised and does not contain a transition of the form $(S,I,I,S)$, then every transition is either 

\begin{itemize}
	\item a movement transition $(S,I,D,NS)$
  \item an output transition  $(S,I,O,NS)$ where there is no transition for $NS,O$ 
  \item an output transition  $(S,I,O,NS)$ where the transition for $NS,O$ is a movement transition
\end{itemize}

We construct $M'$ by transforming each transition in $M$ as follows.  

\begin{itemize}
	\item For each movement transition $(S,I,D,NS)$ in $M$, there is a transition $(S,I,I,D,NS)$ in $M'$
  \item For each output transition $(S,I,O,NS)$ in $M$ where there is no transition for $NS,O$ in $M$, there is a transition $(S,I,O,r,z)$ in $M'$
  \item For each output transition $(S,I,O,NS)$ in $M$ where the transition for $NS,O$ in $M$ is a movement transition $(NS,O,D,NS1)$, there is a transition $(S,I,O,D,NS1)$ in $M'$. 
\end{itemize}

Note that in the third case there will also be a transition $(NS,O,O,D,NS1)$ in $M'$ resulting from the first case. 
\qed
\end{proof}

Taken together, Propositions~\ref{prop:4nocrap} and \ref{prop:4to5} show that for every normalised quadruple machine that terminates on the blank input, there is a quintuple machine that is productivity equivalent. Hence any machines relevant for the busy beaver will be found by searching amongst the quintuple machines alone. In other words, the quadruple machines will not contribute more to the busy beaver problem than the quintuple ones will. 

This of course does not rule out the possibility that there are some intriguing machines to be found amongst the quadruple machines. 
However, it seems unlikely that there is a productivity-equivalent 5-state 2-symbol quadruple machine machine for any 5-state 2-symbol quintuple machine that terminates on the blank input. Clearly it is possible to find {\em some} productivity-equivalent quadruple machine for any 5-state 2-symbol quintuple machine (Ross in fact gives one such transformation), but it seems unlikely that there is a productivity-equivalent quintuple machine \textit{with only 5 states and 2 symbols}. It is an item of future work to settle this issue, either by providing such a transformation or showing that it is impossible. For now, we note that searching quintuple machines subsumes searching for quadruple ones. 

\section{Normal Form}
\label{sec:normal}

In order to solve the busy beaver problem, we are only interested in evaluating machines on the blank input, which means we are able to reduce the number of machines that require non-trivial analysis. In this section we establish the results that show this. To begin with, it is simple to establish the following results, whose proofs are trivial. 

\begin{lem}
\label{lemma:a01}
Let $M$ be a Turing machine containing a tuple of the form $(a,0,\_,\_,a)$. Then the activity of $M$ is $\infty$. 
\end{lem}

\begin{lem}
\label{lemma:a01a}
Let $M$ be a Turing machine containing a tuple of the form $(a,0,\_,\_,z)$. Then the activity of $M$ is $1$. 
\end{lem}

Hence we need only consider machines whose first transition is of the form $(a,0,\_,\_,b)$. It should be noted that we can assume that the second state used in the machine is $b$ (rather than say $c$ or $d$), as justified by the following simple lemma. 

\begin{lem}
\label{lemma:a02}
Let $M$ be a $k$-halting $n$-state $m$-symbol Turing machine containing a tuple of the form $(a,0,\_,\_,S)$ where $S \not \in \{a,b,z\}$. Then there is another $k$-halting $n$-state $m$-symbol Turing machine $M'$ containing the tuple $(a,0,\_,\_,b)$ such that $M$ and $M'$ are productivity and activity equivalent.
\end{lem}

\begin{proof}
Let $M'$ be the machine found by swapping all occurrences of $S$ and $b$ in $M$. The result then trivially follows. 
\qed
\end{proof}

This means that we can insist that the second state encountered in the machine (after the start state $a$) is $b$. Similarly we can insist that the third state (if any) encountered is $c$, and so on. In particular it is straightforward to show the result below. 

\begin{lem}
\label{lemma:statename}
Let $M$ be a Turing machine, and let $S_1, S_2$ be two distinct states in $M$ such that $S_1, S_2 \not \in \{a,z\}$. Let $M'$ be the machine obtained by swapping the states $S_1$ and $S_2$ in every transition in $M$. Then $M'$ is productivity and activity equivalent to $M$ and $M'$ contains the same number of state, symbols and halting transitions as $M$. 
\end{lem} 

It is similarly straightforward to show a similar result for symbols. 

\begin{lem}
\label{lemma:symbolname}
Let $M$ be a Turing machine, and let $O_1, O_2$ be two distinct symbols in $M$ such that $O_1 \not = 0$ and $O_2 \not = 0$. Let $M'$ be the machine obtained by swapping the symbols $O_1$ and $O_2$ in every transition in $M$. Then $M'$ is productivity and activity equivalent to $M$ and $M'$ contains the same number of state, symbols and halting transitions as $M$. 
\end{lem} 

Taken together, Lemmas~\ref{lemma:statename} and ~\ref{lemma:symbolname} mean that we can insist on a specific order in which the states and symbols appear in the execution of the machine. In particular, given that the first state must be $a$, we can insist that the second state encountered be $b$, the third one $c$ and so forth. Similarly, we can insist that the first non-blank symbol countered be 1, the second 2, and so on. In addition, as we know that the second step executed will always be the $b,0$ transition, we can insist that in any transition of the form $(b,0,O,D,S)$ we have $O \in \{0,1,2\}$. 

This means that we can assume that the first transition is of the form $(a,0,O,D,b)$ for some output $O$ and some direction $D$. Now if $O$ is blank, ie the transition is of the form $(a,0,0,D,b)$, then either the tape remains blank throughout the entire execution of the machine, or there is a transition $(s,0,O,D,NS)$ where $s \not = a$ and $O \not = 0$, in which case we simply swap $a$ and $s$. This leads us to the result below. 

\begin{lem}
\label{lemma:a03}
Let $M$ be a $k$-halting $n$-state $m$-symbol Turing machine of finite activity and productivity $\geq 1$ containing a tuple of the form $(a,0,0,\_,NS)$ where $NS \not = z$. Then there is another $k$-halting $n$-state $m$-symbol Turing machine $M'$ of finite activity containing the tuple $(a,0,O,\_,\_)$ where $O \not = 0$ such that $M'$ is productivity equivalent to $M$.
\end{lem}

\begin{proof}
As $productivity(M) \geq 1$, there must be a transition of the form $(S,0,O,\_,\_)$ $S \not = a$ and $O \not = 0$ in $M$, and that this is the first transition in the execution of $M$ which writes a non-blank symbol on the tape. 
Let $M'$ be the machine found by swapping all occurrences of $a$ and $S$ in $M$. The result then trivially follows. 
\qed
\end{proof}

Note that a similar property will follow for a machine of activity $\infty$; 
however, this case is uninteresting for the busy beaver problem. Note also that the activity of $M'$ will be less than that of $M$, as this change effectively ignores the initial execution steps which do not change the blank tape. 

For example, consider the machine in Figure~\ref{fig:m82raw}, which is Dragon 82, a 4-state 3-symbol machine of activity 250,096,776 and productivity 15,008. 
Normalising this machine as above (by swapping states $a$ and $b$) gives the machine in Figure~\ref{fig:m82refine1} which has activity 250,096,775 and productivity 15,008. 

\begin{figure}
\centering
\includegraphics[width=0.7\textwidth]{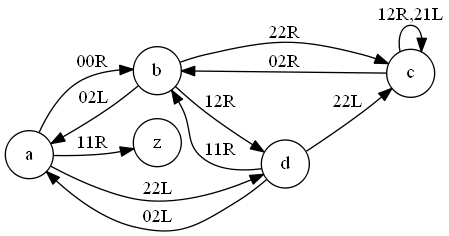}
\caption{Dragon 82}
\label{fig:m82raw}
\end{figure}

\begin{figure}
\centering
\includegraphics[width=0.7\textwidth]{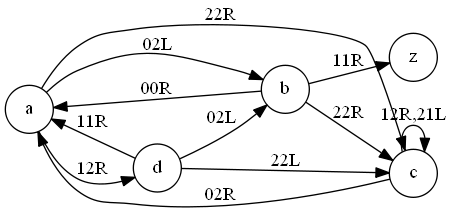}
\caption{Dragon 82 transformed}
\label{fig:m82refine1}
\end{figure}

We can perform a similar transformation under some other circumstances, as specified in the following lemma. Note that as machines of productivity 0 are irrelevant, we only consider those of finite activity and productivity $\geq 1$. 

\begin{lem}
\label{lemma:a03a}
Let $M$ be a $k$-halting $n$-state $m$-symbol Turing machine of finite activity and productivity $\geq 1$ which violates the blank tape condition. Then there is another $k$-halting $n$-state $m$-symbol Turing machine $M'$ of finite activity which is productivity equivalent to $M$ and which satisfies the blank tape condition.
\end{lem}

\begin{proof}
As the execution of $M$ on the blank tape terminates, there must be a final configuration in this execution trace in which the tape is blank. Let $S$ be the state in which this occurs. As the activity of $M$ is finite and the productivity is at least 1, this state cannot be either $z$ or $a$. 
Let $M'$ be the machine found by swapping all occurrences of $a$ and $S$ in $M$. The result then trivially follows. 
\qed
\end{proof}

A similar property will also hold for machines of activity $\infty$.

\begin{figure}[tb]
\centering
\includegraphics[width=0.8\textwidth]{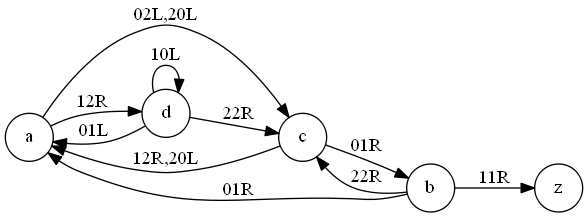}
\caption{Dragon 91 transformation 1}
\label{fig:m91refine1}
\end{figure}
\begin{figure}[tb]
\centering
\includegraphics[width=0.8\textwidth]{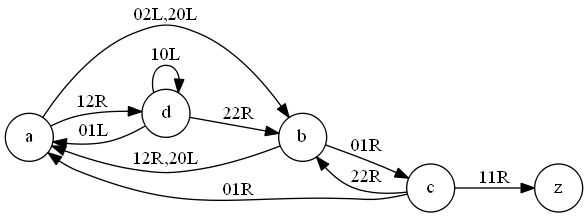}
\caption{Dragon 91 transformation 2}
\label{fig:m91refine2}
\end{figure}

For an example of the application of Lemma~\ref{lemma:a03a}, consider again Dragon~91 in Figure~\ref{fig:m91raw}, which has activity around $10^{14,072}$ and productivity around $10^{7,036}$. 
In the execution of this machine, after 5 steps the machine is in state $b$ and the tape is blank. Swapping $a$ and $b$ gives us the machine in Figure~\ref{fig:m91refine1}. 
We then swap $c$ and $b$ to get the machine in Figure~\ref{fig:m91refine2}.
This machine needs further processing, however, due to the transition $(a,0,2,l,b)$. We can apply Lemma~\ref{lemma:symbolname} to a transition $(a,0,O,\_,S)$ where $O \not \in \{0,1\}$, to ensure that the first transtion must be of the form $(a,0,1,D,b)$.
The choice of $D$, left or right, is entirely arbitrary; as long as this choice is applied consistently, it will have no bearing on the results. We choose $D$ to be $r$, which is consistent with many of the machine definitions on Marxen's website.\footnote{This choice means that our machines are \textit{dextrous}, or right-handed. For every such machine there is a \textit{sinister sibling}, which has exactly the same execution behaviour as the \textit{orthodox original} except that the direction of each transition is reversed.}

We can then transform the machine in Figure~\ref{fig:m91refine2} to one which has initial transition $(a,0,1,r,b)$ by applying both the transformation of Lemma~\ref{lemma:symbolname} and swapping $l$ for $r$ everywhere to get the machine in Figure~\ref{fig:m91refine3}.

\begin{figure}[tb]
\centering
\includegraphics[width=0.7\textwidth]{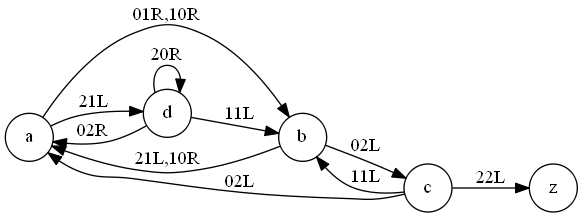}
\caption{Dragon 91 transformation 3}
\label{fig:m91refine3}
\end{figure}
To be strict, we could also insist that the halting transition be $(c,2,1,r,z)$ rather than $(c,2,2,l,z)$ but this is more of a notational convenience than anything else as this changes neither activity nor productivity. In practice, this means that we can restrict our attention to a specific type of halting transition such as $(\_,\_,1,r,z)$ in order to reduce the number of machines that need to be generated, but no essential properties are lost if this constraint is not met, provided that the machine is maximising. Further discussion on this and related issues can be found in Section~\ref{sec:results}. 

This means that the initial transition in all of the machines we consider is $(a,0,1,r,b)$. Given this constraint, we can also deduce some constraints on the second transition used. These are given in the lemma below, whose proof is trivial. 

\begin{lem}
\label{lemma:b01}
Let $M$ be a Turing machine containing tuples of the form $(a,0,\_,r,b)$ and $(b,0,\_,r,S)$ where $S \in \{a,b\}$. Then $M$ has activity $\infty$. 
\end{lem}

This means that we can restrict our attention to transitions of the form $(b,0,\_,l,\_)$ or $(b,0,\_,r,c)$. Furthermore, it is clear that a machine with transitions $(a,0,\_,r,b)$ and $(b,0,\_,\_,z)$ has activity 2, and is hence irrelevant. This observation, together 
with Lemmas \ref{lemma:a01}, \ref{lemma:a02}, \ref{lemma:a03},  \ref{lemma:a03a}, \ref{lemma:statename}, \ref{lemma:symbolname} and \ref{lemma:b01} means that we need only consider machines with a particular initial transition and some constraints on the second one. 

We next show that a particular class of machines (similar to the gutless goanna) is irrelevant to the busy beaver problem. 

\begin{definition}
An $n$-state Turing machine is {\bf $0$-dextrous} if there are $n$ transitions of the form $(\_, 0, \_, \_, \_)$ and all $n$ such transitions are of the form $(\_, 0, \_, r, \_)$.
\end{definition}

Note that a gutless goanna machine is $0$-dextrous. The reason that we identify the $0$-dextrous class of machines is that they are all irrelevant. 

\begin{lem}\label{lemma:0d}
Let $M$ be a $0$-dextrous Turing machine. Then $M$ is irrelevant to the busy beaver problem. 
\end{lem}

\begin{proof}
Let the number of states in $M$ be $n$. If $M$ does not terminate on the blank input, then its activity is $\infty$ and $M$ is hence irrelevant. Otherwise, $M$ terminates on the blank input, which, as $M$ is $0$-dextrous, is only possible if there is a transition in $M$ of the form $(\_,0,\_,r,z)$. This means that $M$ halts in at most $n$ steps, and so $activity(M) \leq n$, which means that $M$ is irrelevant. 
\qed
\end{proof}

This result shows that we can safely ignore any $0$-dextrous machines generated. As this property requires that at least all of the transitions for input $0$ are known, this can generally only be implemented as a constraint on the final machine, i.e.\ that any $0$-dextrous machine that is generated is ignored. 

These results allow us to define an appropriate normal form for machines. Such a definition is given below.

\begin{definition}\label{def:normal}
A Turing machine $M$ is {\bf normal} iff it has all of the following properties:

\begin{itemize}
\item $M$ contains the tuple $(a,0,1,r,b)$
\item $M$ contains either a tuple of the form $(b,0,O,l,S)$ where $S \in \{a,b,c\}$ and $O \in \{0,1,2\}$, or a tuple of the form $(b,0,O,r,c)$ where $O \in \{0,1,2\}$. 
\item $M$ is not $0$-dextrous
\item When executing $M$ on the blank input, 
	\begin{itemize}
		\item states are encountered in alphabetical order
		\item symbols are encountered in numerical order
		\item the blank tape condition is satisfied
	\end{itemize}

\end{itemize}

\end{definition}

The first and second conditions, and the first two parts of the fourth condition, follow from the sequence of results in this section. The third condition follows directly from Lemma~\ref{lemma:0d}. The third part of the fourth condition follows from Lemmas~\ref{lemma:a03} and \ref{lemma:a03a}. 

The constraints on the order of states and symbols exploit Lemmas~\ref{lemma:statename} and~\ref{lemma:symbolname} to eliminate certain types of redundancy. Due to the $(a,0,1,r,b)$ transition, we know that the states $a$ and $b$ and the symbols $0$ and $1$ will be present in every normal machine. This ensures that the third state encountered during computation is $c$, the fourth is $d$ and so on. Clearly there is a productivity and activity equivalent machine in which states $c$ and $d$ are swapped, but in this latter machine, the third state encountered during execution would be $d$, which means this machine is not normal. Similar remarks apply to the symbols, so that in a normal machine, the third symbol encountered will be $2$, the fourth $3$ and so on. One way in which this property becomes important is when comparing machines generated by our process with existing machines. For example, there are some published dreadful dragons which include a transition $(b,0,3,l,a)$. As such a machine is not in normal form, we need to first transform it, which in this case involves swapping the symbols $2$ and $3$, and possibly doing further transformations. More discussion on this point can be found in Section~\ref{sec:results}. 

We are now in a position to show the main result of this section, which relates normal machines to those relevant to the busy beaver problem. 

\begin{prop} \label{prop:relevant}
Let $M$ be an $n$-state Turing machine. If $M$ is relevant to the busy beaver problem, then $M$ is productivity equivalent to a normal machine. 
\end{prop}

\begin{proof}
Let $M$ be a Turing machine which is relevant to the busy beaver problem. 

Then $M$ has finite activity, $activity(M) > n$, $productivity(M) > 0$ and satisfies the blank tape condition. 

Consider the four conditions of Definition~\ref{def:normal}. 

By the combination of Lemmas~\ref{lemma:a01}, \ref{lemma:a01a}, \ref{lemma:a02}, \ref{lemma:statename}, \ref{lemma:symbolname} and~\ref{lemma:a03}, $M$ is productivity equivalent to a machine which satisfies the first condition. 

By the combination of Lemmas~\ref{lemma:statename}, \ref{lemma:symbolname} and~\ref{lemma:b01}, $M$ is productivity equivalent to a machine which satisfies the second condition.  

As $M$ has finite activity, by Lemma~\ref{lemma:0d} it cannot be $0$-dextrous. 

By Lemmas~\ref{lemma:statename} and~\ref{lemma:symbolname}, $M$ is productivity equivalent to a machine which satisfies the first two parts of the fourth condition. 
The third part of the fourth condition follows immediately from the definition of relevance. 
\qed
\end{proof}

This means that in order to solve the busy beaver problem, it is sufficient to consider only normal machines. 
In other words, we can be certain that any machines excluded from consideration are guaranteed to be irrelevant. Note also that the case $O = 2$ in the second condition is only applicable when the number of symbols in the machine is at least 3 and the number of states is at least 3. 

The results of this section allow us to define an appropriate normal form for machines, which means that there are certain machines that need not be generated at all, thus reducing the work we need to do, and others which once generated can be immediately dismissed as irrelevant, thus reducing the number of machines to be stored. This means we are almost in a position to define a procedure to generate machines, but there is one final aspect to consider before we do. 

\section{Monotonicity and Machine Generation}
\label{sec:monotonicity}

As noted above, when generating machines for the busy beaver problem, it seems intuitively natural to concentrate on machines that are \textit{1-halting, exhaustive} and \textit{maximising}. We can ensure that a machine is maximising by specifying that any halting transition be of the form $(\_, \_, 1, r, z)$. 
It is more difficult to ensure that a machine is exhaustive and 1-halting. In fact, we cannot guarantee that a generated machine satisfies either of these properties. 

\begin{figure}
\centering
\includegraphics[width=0.6\textwidth]{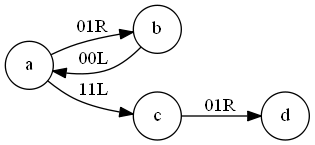}
\caption{Partially generated machine 1}
\label{fig:partial1}
\end{figure}

For example, consider the partial machine in Figure~\ref{fig:partial1} which is generated as part of the process of generating $5$-state $2$-symbol machines.
Execution of this partial machine brings us to the configuration $1\{d\}10$. So what are our choices for the $(d,1,O_1,D_1,N_1)$ transition? The set of states used in the machine so far is $\{a,b,c,d\}$ and the set of symbols used so far is $\{0,1\}$. This means that we have used all the symbols necessary for a 2-symbol machine, but have used only 4 states. Hence we should consider only the possibilities $O_1 \in \{0,1\}$, $D_1 \in \{l,r\}$ and $N_1 \in \{a,b,c,d,e\}$. Note that $z$ is excluded from the possible states for this transition, as we are generating a 5-state machine and have only used $\{a,b,c,d\}$ so far, and so if we allow $z$ as a possible resulting state from this transition, we will have generated a terminating 4-state 2-symbol machine, and not a 5-state 2-symbol one. Let us assume that we choose $O_1 = 1, D_1 = r$ and $N_1 = e$. This gives us the machine in Figure~\ref{fig:partial2} in the configuration $11\{e\}0$.

\begin{figure}
\centering
\includegraphics[width=0.8\textwidth]{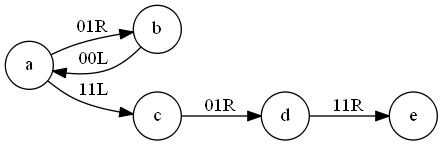}
\caption{Partially generated machine 2}
\label{fig:partial2}
\end{figure}

Our next decision is for the transition $(e,0,O_2, D_2, N_2)$. Now as we have used 5 states and 2 symbols in the definition so far, one possibility is to choose this transition to be the halting transition, in which case we have $O_2 = 1, D_2 = r$ and $N_2 = z$. Alternatively we may choose $O_2 \in \{0,1\}, D_2 \in \{l,r\}$ and $N_2 \in \{a,b,c,d,e\}$ and continue the process. It may seem counterproductive to choose the halting transition at this point, as we have a number of alternatives to it. 
\textbf{The problem is that we have no way to guarantee that any of these alternatives will result in a terminating machine.} What we do know is that the choice of the halting transition at this point \textit{will} result in a terminating machine. Hence, despite the apparent redundancy, it seems the only safe course at this point is to output the machine in Figure~\ref{fig:partial3} as one possibility, and continue to search for more. This means that we can generally only guarantee that the generation of an $n$-state $m$-symbol machine will result in a machine that is  \textbf{$n$-state full and $m$-symbol full}, rather than one which is $n$-$m$-exhaustive. 

\begin{figure}
\centering
\includegraphics[width=0.8\textwidth]{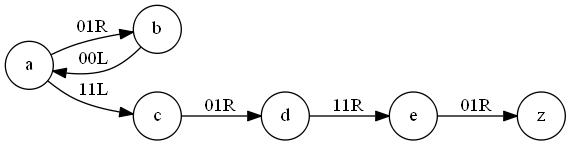}
\caption{Partially generated machine 3}
\label{fig:partial3}
\end{figure}

The reason that we can always ensure that the generation of an $n$-state $m$-symbol machine results in a machine which is $n$-state full and $m$-symbol full is \textit{the strict monotonicity of the busy beaver function.} In other words, if we know that  that $bb(n_1, m) > bb(n_2, m)$ whenever $n_1 > n_2 \geq 2$, then this justifies the above step in which the halting transition was not considered when searching for the $(d,1,O_1,D_1,N_1)$ transition for the machine in Figure~\ref{fig:partial2}. Specifically, knowing that $bb(5,2) > bb(4,2)$ means that we are justified in not considering the halting transition until there are at least 5 states in the machine generated. In general, this means that when generating an $n$-state $m$-symbol machine, we do not allow the halting transition to be added until we have all $n$ states and all $m$ symbols present in the machine. As noted above, the transition $(a,0,1,r,b)$, which occurs in all machines, guarantees the occurrence of at least 2 states and at least 2 symbols. 

It is intuitively obvious that the busy beaver function is monotonic in both the number of states and the number of symbols, i.e.\ that  that $bb(n_1, m) \geq bb(n_2, m)$ whenever $n_1 > n_2 \geq 2$ and that $bb(n, m_1) \geq bb(n, m_2)$ whenever $m_1 > m_2 \geq 2$, as any machine with at most $n$ states (respectively $m$ symbols) clearly has at most $n+1$ states (respectively $m+1$ symbols). 

It is not difficult to show that the $bb$ function (and $\mathit{ff}$ for that matter) is strictly monotonic in the number of states. This is done in the Proposition below, which is a straightforward generalisation of Example 4.5 in \cite{BBJ}. 

\begin{prop}\label{prop:states}
Let $M$ be a $k$-halting $n$-state $m$-symbol Turing machine with finite activity. Then there is a $k$-halting $(n+1)$-state $m$-symbol Turing machine $M'$ with finite activity such that $activity(M') \geq activity(M) + 1$ and $productivity(M') = productivity(M) + 1$. Furthermore, if $M$ is $n$-$m$-exhaustive, then $M'$ is $(n+1)$-$m$-exhaustive. 
\end{prop}

\begin{proof}
As $M$ terminates on the blank input, there must be a halting transition of the form $(S, I, O, D, z)$ in $M$. Let $s$ be a state which does not occur in $M$. Consider the machine $M'$ which is obtained from $M$ by replacing the transition $(S, I, O, D, z)$ with $m$ transitions as below. A diagrammatic representation of this transformation is in Figure~\ref{fig:transform1}.

\begin{center}
\begin{tabular}{ccccc}
\textbf{State} & \textbf{Input} & \textbf{Output} & \textbf{Direction} & \textbf{New State} \\
S & I & O & D & s \\
s & 0 & 1 & r & z \\
s & 1 & 1 & r & s \\
\ldots & & & & \\
s & m-1 & m-1 & r & s \\
\end{tabular}
\end{center}
 
Note that once $M'$ enters state $s$, it will skip to the right until it comes across a $0$, and then halt.
So the execution of $M'$ on the blank input will behave exactly as $M$ until the halting transition of $M$ occurs, at which point $M'$ will change to the new state $s$ and will execute at least one more step than $M$ before terminating. Note also that when $M'$ enters state $s$, the tape will be in exactly the same configuration as when $M$ terminates, apart from being in state $s$ rather than state $z$. This means that when $M'$ terminates, it will change a 0 into a 1, and hence have productivity one more than that of $M$. 

For the exhaustiveness property, note that one transition in $M$ is replaced with $m+1$ transitions in $M'$, and so if $M$ contains $n \times m$ transitions, then $M'$ contains $(n \times m) - 1 + (m+1) = (n+1) \times m$ transitions. 
\qed
\end{proof}

An example of this transformation is given in Figure~\ref{fig:exampletransform}.

\begin{figure}[tb]
\centering
\begin{tabular}{cc}
\includegraphics[width=0.3\textwidth]{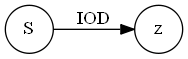} & \includegraphics[width=0.55\textwidth]{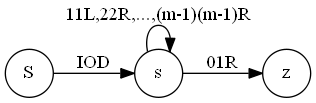} \\
Original machine & Transformed machine \\
\end{tabular}
\caption{Machine transformation}
\label{fig:transform1}
\end{figure}

\begin{figure}[tb]
\centering
\begin{tabular}{cc}
\includegraphics[width=0.4\textwidth]{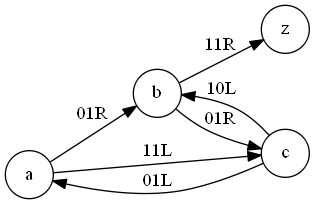} & \includegraphics[width=0.5\textwidth]{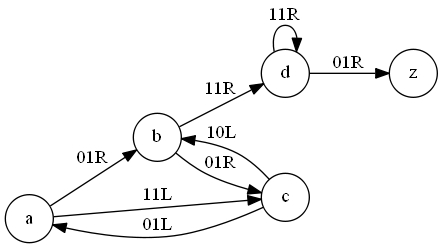} \\
Original machine & Transformed machine \\
\end{tabular}
\caption{Example transformation}
\label{fig:exampletransform}
\end{figure}

This transformation is straightforward, and this method of increasing machine size is highly unlikely to be of any practical use in finding busy beaver values. However, it establishes the strict monotonicity of the busy beaver function in the number of states, which despite being virtually the weakest possible statement of strict monotonicity, is sufficient to ensure that the above procedure for adding the halting transition is sound. 

It seems intuitively clear that a similar result should hold for the number of symbols, i.e.\ that $bb(n, m_1) > bb(n, m_2)$ whenever $m_1 > m_2 \geq 2$, and so that when generating an $m$-symbol machine, we require that at least $m$ symbols be present in the machine before we add a halting transition. As noted above, the presence of the transition $(a,0,1,r,b)$ in every machine generated ensures that every generated machine contains at least 2 symbols. 

A formal statement  of this desired result is given below. 

\begin{conj}\label{conj:symbols} 
Let $M$ be a $k$-halting Turing machine with $n$ states and $m$ symbols for some $k \geq 1$ with finite activity. Then there is a $k$-halting $n$-state $(m+1)$-symbol Turing machine $M'$
with finite activity such that $activity(M') > activity(M)$ and $productivity(M') > productivity(M)$.
\end{conj}

\begin{figure}[tb]
\centering
\begin{tabular}{cc}
\includegraphics[width=0.5\textwidth]{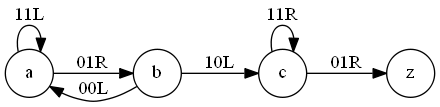} & \includegraphics[width=0.5\textwidth]{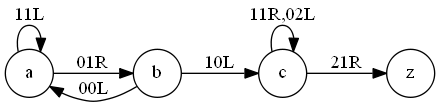} \\
\end{tabular}
\caption{Machine 3}
\label{fig:machine3}
\end{figure}

Unfortunately we have been unable to prove this result, despite it seeming to be obviously true. The difficulty for the proof is to find a similar transformation to the one in the proof of Proposition~\ref{prop:states}. In particular, it seems problematic to come up with a transformation that will take a terminating computation in an $m$-symbol machine, and transform it into a terminating computation in an $(m+1)$-symbol machine. To see the difficulty, consider the machine in the left hand side of Figure~\ref{fig:machine3}.
This has activity 7 and productivity 2. 
It is not hard to tweak this machine by changing the halting transition from $(c,0,1,r,z)$ to $(c,0,2,l,c)$, and adding a new halting transition $(c,2,1,r,z)$. This machine is given in the right hand side of Figure~\ref{fig:machine3}. This makes the execution as below.

\begin{center}
$\{a\}0 \Rightarrow 1\{b\}0 \Rightarrow \{a\}1 \Rightarrow \{a\}01 \Rightarrow 1\{b\}1 \Rightarrow \{c\}1 \Rightarrow 1\{c\}0 \Rightarrow \{c\}12 \Rightarrow 1\{c\}2 \Rightarrow 11\{z\}$
\end{center}

This has activity 9, which is an increase on the previous 7, but still has productivity only 2. 
We may of course consider other changes to the machine, but it only seems safe to change the halting transition, and to add transitions of the form $(\_,2,\_,\_,\_)$, as changing any of the other 5 transitions means that we will not be able to guarantee that the execution of this machine still terminates. Finding some appropriate transformation and hence providing a proof of Conjecture~\ref{conj:symbols} remains an item of future work.

As mentioned above, an even stronger result is desirable here, i.e.\ that we can guarantee that all generated machines are exhaustive. To do so would require a result similar to Proposition~\ref{prop:states}, but showing the strict monotonicity of the busy beaver function in terms of the number of transitions in the machine. As Chaitin has argued \cite{Chaitin87}, it may be more natural to stratify the busy beaver values by classifying them in terms of the number of transitions in the machine, rather than the number of states or symbols used in its definition. Whilst this seems intuitively appealing, the generation of exhaustive machines can only be guaranteed by proving this form of strict monotonicity, which seems significantly more difficult to show than the simple proof given above. Finding and proving such a result will presumably require a much deeper understanding of the nature of terminating busy beaver machines than we have at present. 

It should also be noted that we cannot guarantee that a generated machine will be 1-halting, as we cannot guarantee that a partially generated machine will always terminate, and so executing the machine may never result in an opportunity to add a halting transition. This means that we sometimes have to settle for a machine which is 0-halting, due to the way in which machines are generated by executing partially defined machines. Further discussion on this point is deferred until Sections~\ref{sec:generation} and \ref{sec:results}.

\section{Generation process}\label{sec:generation}

We are now in a position to define the process for generating machines with $n$ states and $m$ symbols where $n,m \geq 2$. This will follow the same general strategy as the \textit{tnf} process described by Lin and Rado \cite{LR65}, but refined in the light of the above results.

We refer to the states in the transitions defined in $M$ as \textit{states}$(M)$, and to the symbols in the transitions defined in $M$ as \textit{symbols}$(M)$. 
For convenience, we will refer to the states as $\{a_1, a_2, \ldots a_n\}$ where $a_1 = a$ and $a_2 = b$. 
We denote by $state\_choice(n, \{a_1, \ldots a_k\})$ the set $\{a_1, \ldots a_k, a_{k+1}\}$ if $k < n$ and $\{a_1, \ldots a_n\}$ otherwise. 
Similarly we denote by $symbol\_choice(m, \{0,1,\ldots l\})$ the set $\{0,1,\ldots l, l+1\}$ if $l < m-1$ and $\{0,1,\ldots m-1\}$ otherwise. 
For example, $state\_choice(4, \{a_1, a_2\}) = \{a_1,a_2,a_3\}$, which will be used to ensure that the next state chosen after $a_1$ and $a_2$ is $a_3$.
Similarly $symbol\_choice(4, \{0,1\}) =\{0,1,2\}$, which will be used to ensure that the next symbol chosen after 1 is 2.  

Note that the procedure $generate(n,m)$ below is non-deterministic; to find all appropriate machines, we need to exhaustively search all ways of outputting a machine from it. 

The procedure $generate(n,m)$ is defined as below. \\

$generate(n,m)$ 
\begin{enumerate}
\item Initialise the machine $M$ to the tuple $(a,0,1,r,b)$.
\item Choose the $b,0$ transition satisfying one of the conditions below and add this transition to $M$.   
		\begin{itemize}
		\item $(b,0,O,l,NS)$ where $NS \in \{a,b\}$, $O \in symbol\_choice(m, \{0,1\})$ 
		\item	if $n \geq 3$, $(b,0,O,D,c)$ where $D \in \{l,r\}$, $O \in symbol\_choice(m, \{0,1\})$  
    \end{itemize}
\item Execute $M$ on the blank input until either 		
			\begin{itemize}
				\item $M$ is known to be irrelevant, or the bound on the number of execution steps is exceeded. Output $M$ and halt. 
				\item an undefined combination of state $S$ and input $I$ is found. 
			\end{itemize}
\item Choose a new transition $(S,I,O,D,NS)$ for $M$ as follows. 
\begin{itemize}
	\item If $M$ is $n$-state full and $m$-symbol full, and $M \cup \{(S,I,1,r,z)\}$ is not 0-dextrous \\
	      $~~~$ add $(S,I,1,r,z)$ to $M$, output $M$ and halt
	\item If $M \cup \{(S,I,O,D,NS)\}$ is not 0-dextrous \\
	      $~~~$ add $(S,I,O,D,NS)$ to $M$ \\
				$~~~~~$ where $NS \in state\_choice(n,states(M))$, $O \in symbol\_choice(m,symbols(M))$, $D \in \{l,r\}$
\end{itemize}
\item If $|M| = n \times m - 1$  add $(S,I,1,r,z)$ to $M$ for the appropriate $S$ and $I$, output $M$ and halt. 
\item Go to step 3.
\end{enumerate}

Note that we halt the process whenever a halting transition is added to $M$. 

Steps 1 and 2 are derived directly from Proposition~\ref{prop:relevant}. 

Step 3 us where the machine is executed until either the machine is known to be irrelevant (such as halting with activity $\leq n$, or violating the blank tape condition), exceeds a given bound on computation length, or finds a place in the machine where a new transition is needed. Clearly we need to store the machines whose computation length which exceeds the given bound, but strictly speaking, we could insist that machines that are known to be irrelevant are not stored at all. Whilst this would reduce the number of machines considered, we have chosen to retain such machines in order to simplify the definition of the generation process. The only exception to this rule is that we choose not to store $0$-dextrous machines, due to the fact that we can easily check whether a machine is $0$-dextrous or not from its definition alone (i.e.\ without executing the machine). However this is very much a matter of taste rather than anything of great significance. 

Step 4 is based on Proposition~\ref{prop:states} and Conjecture~\ref{conj:symbols}. This means that a halting transition is only considered if the partial machine generated already contains $n$ states and $m$ symbols, i.e.\ there are no unused states or symbols. Otherwise, Proposition~\ref{prop:states} and Conjecture~\ref{conj:symbols} indicate that generate a machine of larger activity. The second bullet point in Step 4 ensures that states and symbols are added in the appropriate order, so that if for example the partial machine contains states $\{a,b,c\}$ and symbols $\{0,1\}$, then the next state (if any) to be added is $d$, and the next symbol (if any) is $2$. 

Step 5 ensures that if we get to a point where there is only one possible transition to add, then the halting transition is added, and as the machine is fully defined, there is no need for any further execution. For example, consider generating a 5-state 2-symbol machine, in which there are 8 transitions, leaving the only two unspecified transitions as those for $e,0$ and $d,1$. If in Step 3 find that we need a transition where $S = e$ and $I = 0$, then there are two possibilities. One is to add the halting transition as the one for $e,0$ (Step 4, first bullet point). This is because we already have all states $\{a,b,c,d,e\}$ occurring in transitions in the machine, as well as the symbols $\{0,1\}$. Another possibility is to add a non-halting transition for $e,0$ (Step 4, second bullet point), at which point we will have 9 non-halting transitions in the machine. As there is only one remaining unspecified transition (the one for $d,1$), this must be the halting transition, and so we add the transitions $(d,1,1,r,z)$ and halt. 

Note that in Step 4 we include a test to ensure that the machine generated is not 0-dextrous. By Lemma~\ref{lemma:0d} we know that such machines are irrelevant, and so there is no need to generate such machines. Note also that we cannot guarantee that any machine generated by this process is exhaustive; as noted above, the best we can do seems to be a guarantee that it is $n$-state full and $m$-symbol full for the appropriate $n$ and $m$. 

A further issue arises from Step 3, which is that it is possible that the machine generated so far may not terminate. In order to deal with this issue it seems sensible to incorporate some kinds of non-termination check into the execution process. This might include checking whether the current configuration has occurred earlier, based on the history of the execution trace, or possibly more sophisticated techniques \cite{harland07}. At the very least, it would seem prudent to include an upper bound on the number of execution steps permitted. A further check is suggested by Lemma~\ref{lemma:a03}, which is that we should ignore any machine in which a blank tape occurs other than in the initial configuration. As in the Lemma, a terminating machine in which a blank tape occurs for a state other than $a$ or $z$ will have an equivalent machine without this property. A machine for which this occurs in state $z$ has productivity 0 and is hence irrelevant. A machine for which this occurs in state $a$ other than the initial configuration has activity $\infty$ and is hence irrelevant, as is any other kind of non-terminating machine. Hence if we find that the configuration is blank at any point other than the initial configuration, we should cease execution and note the machine generated as irrelevant. 

This means that there will be some machines generated whose status is known immediately. These include any non-terminating machines detected in Step 3, as well as machines generated by the first bullet point in Step 4, which are known to terminate. So there will be some machines that are already classified as they are generated, including some (the ones found in Step 3 to be non-terminating) which are 0-halting, whilst many of the other machines generated will be unclassified (those whose execution exceeds the bound in Step 3, and those found in Step 5). This somewhat messy arrangement seems to be an unavoidable consequence of the use of the \textit{tnf} technique. In principle, we could add halting transitions to the 0-halting machines so that the set of machines would appear more uniform. However, this seems to be needlessly complicated, and is not difficult to perform if it is required later for some reason. 

A further issue that arises from 0-halting machines is that it is possible that some of these machines may need to be revisited at a later point. As noted above, some 0-halting machines are generated due to the execution bound being exceeded during the generation process. It is possible that later analyses will show that these machines do not terminate. However, if that subsequent analysis cannot do this, then we need to consider the possibility that these machines may need to be defined in more detail, as it may be that these machines can be extended to machines which will terminate on the blank input.  For example, consider the 4-state 2-symbol machine in Figure~\ref{fig:0halting}. This machine exceeded the bound, and so did not have a halting transition added.

\begin{figure}
\centering
\includegraphics[width=0.7\textwidth]{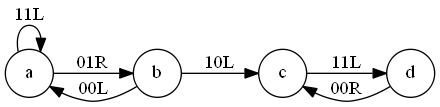}
\caption{0-halting machine}
\label{fig:0halting}
\end{figure}

After 5 steps of computation, this machine is in the configuration  $0\{c\}1$, and it is not hard to see that this machine will oscillate between the configuration $0\{c\}1$ and the configuration $\{d\}01$, and so any extension of this machine will have the same non-terminating behaviour. However, it is not clear that this property will always hold, and so we may need to reconsider such 0-halting machines at a later point. An example of this behaviour is given in the next section.

\section{Implementation and Results}\label{sec:results}

We have implemented the above procedure in SWI-Prolog (version 7.2.3, multi-threaded, 64-bits) \cite{swi}. This is part of a suite of around 5,000 lines of code, with the part dealing specifically with the generation process (when separated from the execution of machines) being only around 200 lines of code. We have used this procedure to generate machines of dimension 4, 6, 8, 9 and 10 (Table~\ref{table:machines}). As in \cite{Harland16}, we use the term \textbf{Blue Bilby} for machines of dimension up to 6, \textbf{Ebony Elephant} for those of dimension 8, \textbf{White Whale} for those of dimension 9 or 10, and \textbf{Demon Duck of Doom} for those of dimension 12. All of our results were obtained on a PC running Windows 7 with an Intel i7 3.6 GHz processor, 8 GB of RAM and a disk capacity of 500 GB. 
All of the code and data referred to here are available at the author's website at \url{www.cs.rmit.edu.au/~jah/busybeaver}.

For comparison, we have also implemented two further processes for generating machines, which will only be possible for the smaller classes. 
This will provide a means of analysing the effectiveness of the tree normal form procedure and a basis for debugging if need be. 
The first of these, which we refer to as \textit{free} generation, generates 1-halting machines by using the same set of $a,0$ and $b,0$ transitions as in the tree normal form process, but otherwise the transitions are generated arbitrarily. The second, which we refer to as \textit{all} generation, generates 1-halting machines arbitrarily, without any restrictions on either the $a,0$ or $b,0$ transitions. In both of these cases, the generation of a 1-halting $n$-state $m$-symbol proceeds until the machine contains $(n \times m) - 1$ transitions, at which point the halting transition is added and the machine is complete. Note that the \textit{free} machines are precisely the \textit{all} machines in which the first transition is $(a,0,1,r,b)$ and the $b,0$ transition obeys the same restrictions as in Step 2 of procedure $generate(n,m)$. 
Both of these methods, as shown above, will generate redundant machines, and we cannot expect these to be as effective as the tree normal form process. However, doing so will hopefully give us some guidance in the cases when we only have available the \textit{tnf} machines.

For example, Table~\ref{table:machines} shows that there are 2,148,483,648 machines for the $4 \times 2$ \textit{all} case, compared to 50,311,648 for the $4 \times 2$ \textit{free} case and 511,145 from the tree normal form process. This shows that tree normal form generation in this case reduces the number of machines to be considered by a factor of over 4,000 compared to \textit{all} generation, and by a factor of just under 100 over \textit{free} generation. For the $2 \times 4$ case, there are 37,748,736 and 342,516 machines in the \textit{free} and \textit{tnf} cases respectively, which means the tnf process reduces the number of machines by a factor of over 6,000 compared to \textit{all} generation, and by a factor of a little over 100 over \textit{free} generation. 

Note also that whilst we have not explicitly generated the $2 \times 4$ machines in the \textit{all} case, we know that there will be exactly the same number of these machines as for the $4 \times 2$ \textit{all} case. We can also generate these from the $4 \times 2$ machines if need be. As all possible $4 \times 2$ machines have been generated, this will include all possible instances of transitions $(S,I,O,D,NS)$, where $S, NS \in \{a,b,c,d\}$ and $I,O \in \{0,1\}$. By swapping states and symbols, this same set of machines can be easily transformed into a set of transitions that includes all possible instances of $(I,S,NS,D,O)$, and by renaming $a,b,c$ and $d$ to $0,1,2$ and $3$ respectively and $0$ and $1$ to $a$ and $b$ respectively, we can generate the corresponding $2 \times 4$ machines. Whilst it seems plausible that this will be faster than generating all such machines from scratch, we have not confirmed this. 

\begin{table}
\centering
\begin{tabular}{|l|c|r|r|r|r|r|r|}\hline
\textbf{Class}          & \textbf{Size} & \textbf{Tnf}  & \textbf{Time (s)} & \textbf{Free} & \textbf{Time (s)}  & \textbf{All}   & \textbf{Time (s)} \\ \hline
\textit{Blue Bilby}     & $2 \times 2$  & 36            & 0.04              & 64            & 0.01               & 2,048          & 0.08              \\
                        & $3 \times 2$  & 3,508         & 0.88              & 55,296        & 2.75               & 1,492,992      & 74.31             \\
                        & $2 \times 3$  & 2,764         & 0.79              & 41,472        & 2.04               & 1,492,992      & 72.90             \\ \hline
\textit{Ebony Elephant} & $4 \times 2$  & 511,145       & 196.11            & 50,331,648    & 3,039.24           & 2,148,483,648  & $\approx$ 170,000 \\
                        & $2 \times 4$  & 342,516       & 145.61            & 37,748,736    & 2,271.52           & 2,148,483,648  &  --               \\ \hline
\textit{White Whale}    & $3 \times 3$  &  26,813,197   & 10,784.62         &  --           &  --                & --             &  --               \\
                        & $5 \times 2$  & 102,550,546   & 78,490.98         &  --           &  --                & --             &  --               \\
							          & $2 \times 5$  &  75,402,497   & 48,399.56         &  --           &  --                & --             &  --               \\ \hline
\textit{Demon Duck}     & $6 \times 2$  & $ \geq 1,540,000,000$   & $\geq$ 1,200,000  &         &                    &                &               \\ \hline
\end{tabular}
\caption{Number of machines}
\label{table:machines}
\end{table}

\begin{table}
\centering
\begin{tabular}{|c|c|c|c|c|c|c|c|c|c|}\hline
                 & \textbf{2 $\times$ 2} & \textbf{2 $\times$ 3} & \textbf{2 $\times$ 4} & \textbf{2 $\times$ 5}   & \textbf{3 $\times$ 3} & \textbf{3 $\times$ 2} & \textbf{4 $\times$ 2} & \textbf{5 $\times$ 2} & \textbf{6 $\times$ 2}\\ \hline
\textbf{Total}   & 36             & 2764           & 342,516        & 75,402,497       & 26,813,197   & 3508           & 511,145        & 102,550,546    & \textbf{??} \\  \hline
$b,0,0,l,a$      & 25.0\%         &  6.0\%         &  3.2\%         &  2.3\%           &  3.2\%       &  6.5\%         &  3.6\%         &  2.5\%         &        503,314,910 \\		         
$b,0,1,l,a$      & 25.0\%         & 17.2\%         &  9.3\%         &  6.1\%           & 10.5\%       & 18.7\%         & 12.4\%         &  8.7\%         & $\geq$ 1,036,685,090 \\
$b,0,2,l,a$      &                & 35.2\%         & 50.1\%         & 57.4\%           & 15.7\%       &                & 		            &                & \\
$b,0,0,l,b$      & 25.0\%         &  6.5\%         &  4.0\%         &  2.7\%           &  3.0\%       &  6.5\%         &  3.1\%         &  2.1\%         & \\
$b,0,1,l,b$      & 25.0\%         & 16.4\%         &  9.0\%         &  6.0\%           &  7.0\%       & 13.3\%         &  7.2\%         &  4.9\%         & \\
$b,0,2,l,b$      &                & 18.7\%         & 24.4\%         & 25.5\%           &  7.6\%       &                &                &                & \\
$b,0,0,l,c$      &                &                &                &                  &  5.3\%       &  8.9\%         & 10.5\%         & 11.0\%         & \\
$b,0,0,r,c$      &                &                &                &                  &  3.7\%       &  7.4\%         & 13.5\%         & 16.6\%         & \\
$b,0,1,l,c$      &                &                &                &                  & 12.0\%       & 22.9\%         & 28.0\%         & 29.0\%         & \\
$b,0,1,r,c$      &                &                &                &                  &  7.1\%       & 15.9\%         & 21.7\%         & 25.2\%         & \\ 
$b,0,2,l,c$      &                &                &                &                  & 14.3\%       &                & 		            &                & \\
$b,0,2,r,c$      &                &                &                &                  & 10.6\%       &                & 		            &                & \\ \hline
\end{tabular}
\caption{Number of machines per $b,0$ transition}
\label{table:machines3}
\end{table}

Apart from the $4 \times 2$ \textit{all} and $2 \times 4$ \textit{all} cases, for all classes up to and including \textbf{White Whale} the number of machines in Table~\ref{table:machines} are not overwhelmingly large for a typical modern personal computer. In particular, the tree normal form cases can be stored with relative ease. In our implementation we have chosen to store these machines in plain text files, with 1,000,000 machines per file. This is a simple and convenient means of storage, but not a particularly efficient one. Nevertheless, with judicious use of modern compression tools such as \textit{7-Zip} \cite{7zip}, this is adequate for our purposes. The figure of 1,000,000 seems a reasonable one, but is more or less arbitrary, although a conveniently round figure like this makes it simpler to count the number of machines generated. It may be appropriate at some point to store these machines in a database accessible via the Web, but making these available via compressed text files is presumably adequate for experimental purposes.  

It is hoped that the analysis of the machines up to and including the \textbf{White Whale} will provide some insights that will lead to further reductions that can be made before tackling the \textbf{Demon Duck of Doom}, for which the numbers seem prohibitive at present. In particular, it is hoped that it will be possible to reduce the number of machines that need to be generated by analysing the smaller classes and identifying stronger criteria for relevance that can be translated into significant reductions in the search space, such as requiring a certain sequence of transitions to be present in order to generate large productivities.  

It also seems highly likely that the \textbf{Demon Duck of Doom} is beyond the capabilities of a typical desktop machine, unlike the smaller cases, and hence will require cloud computing methods, both for generating and storing machines, and for their analysis. At the very least, the \textbf{Demon Duck of Doom} will require a significantly greater level of storage. Not only are there four classes of machines to consider ($6 \times 2, 4 \times 3, 3 \times 4$, and $2 \times 6$), but the number of machines in any of these classes is likely to be much greater than we can reasonably expect to store on commodity hardware. We have generated a fraction of the $6 \times 2$ machines, in order to get some indication of the likely number of machines in this class. As shown in the table in Table~\ref{table:machines3}, there are 503,314,910 machines for the first $b,0$ transition (i.e.\ 503,314,910 6-state 2-symbol machines are generated by the tree normal form process with the transition $(b,0,0,l,a)$). This means that there are around 5 times as many $6 \times 2$ machines for the first $b,0$ transition alone as there are for the entire class of $5 \times 2$ machines. In order to get a more precise estimate of the likely number of $6 \times 2$ machines, we have analysed the numbers of machines for each of the possible $b,0$ transitions (Table~\ref{table:machines3}). Given that the transition $(b,0,0,l,a)$ accounted for 3.6\% and 2.6\% of the $4 \times 2$ and $5 \times 2$ machines respectively, it seems likely that the 503,314,910 $6 \times 2$ machines represent at most 1.6\% of the total. This means that the total number of $6 \times 2$ machines is at least 32,000,000,000, and possibly higher. Using our current (rather wasteful) approach of around 180 bytes per $6 \times 2$ machine, this means that it will take at least $32 \times 10^9 \times 180 = 5.24$ TB to store the unclassified $6 \times 2$ machines alone. We can of course use compression tools to reduce this size once all the machines are generated. This suggests that finding a more efficient storage mechanism (possibly via bit-mapping methods) will be an important consideration. It also seems that investigating machines up to and including the \textbf{White Whale} is a natural ``breakpoint'' in the analysis, as it seems necessary to learn as much as possible about the distribution of machines and their classifications before tackling the \textbf{Demon Duck of Doom}. 

We have also included the time taken to generate each class of machines. These range from a few seconds or less to around 2 days for classes up to and including the \textbf{White Whale} (and significantly longer for the fraction of the $6 \times 2$ machines generated so far). It is quite possible that these times can be substantially improved, but our point in recording them here is to show that the \textbf{White Whale} and all smaller classes can be generated in reasonable amounts of time without the need for any special processing or storage arrangements.

As can be seen from Table~\ref{table:machines}, there are generally less $2 \times n$ machines than there are $n \times 2$ machines. This is due to the $b,0$ transition, as if there are only two possible states, there are only 6 possibilities for this transition. However, for the $n \times 2$ case, there are 8, due to the possibility of a third state (which, as discussed above, can be assumed to be $c$). We note that from the same table the ratio of the number of machines in the $2 \times n$ case to the number of machines in the $n \times 2$ case is 78\%, 66\% and 73\% respectively for $n = 3,4,5$. We presume a similar property will hold for $n = 6$, meaning that there will be more $6 \times 2$ machines than $2 \times 6$ ones, i.e.\ we expect that if the above estimation of the number of $6 \times 2$ machines is correct, then there are ``only'' around 24,000,000,000 $2 \times 6$ machines. 

Having generated the various classes of machine up to and including the \textbf{White Whale}, it seems natural to ask where the known dreadful dragons \cite{Harland16} are to be found in these classes. Of the 100 machines evaluated in \cite{Harland16}, 50 of these are of dimension 10 or less, and hence will be located somewhere in the machines generated. We have located these 50 machines in our list of machines, as noted in Table~\ref{table:dragons}. The numbering is the same as used in \cite{Harland16}. 

\noindent
The \textbf{No.} entry is the given number of the machine from \cite{Harland16}. \\
The \textbf{Dim.} entry is the dimension of the machine. \\
The \textbf{Identifier} entry is our identifier for the machine, based on the sequence in which the machines were generated. \\
The \textbf{Productivity} entry is the productivity of the machine. \\
The \textbf{Incomplete} entry indicates whether or not the machine was completely defined in our search. \\
The \textbf{Modified} entry indicates whether or not it was necessary to modify the published definition of the machine to fit our restrictions. \\
The final two columns contain the same information as the previous \textbf{No.} and \textbf{Identifier} entries, but are sorted by \textbf{Identifier} rather than \textbf{No.} within each class of machines. These are intended to show the clustering, if any, of dragons around particular places in the search. 

\begin{table}
\centering
\begin{tabular}{llrrcc|lr}
\textbf{No.} & \textbf{Dim.} & \textbf{Identifier} & \textbf{Productivity} &  \textbf{Incomplete} & \textbf{Modified} & \textbf{No.} & \textbf{Identifier} \\
 1 & 2x4 &    111,941  & 84             & & Yes & 1 &  111,941  \\
 2 & 2x4 &    112,118  & 90             & &     & 2 &  112,118 \\
 3 & 2x4 &    112,118  & 90             & & Yes & 3 &  112,118  \\
 4 & 2x4 &    229,000  & 2,050          & & Yes & 4 &  229,000   \\ \hline

 5 & 3x3 &  1,797,985  & 31             & & & 17 &     152,890  \\
 6 & 3x3 & 21,115,998  & 5,600          & & & 10 & 	 1,576,778  \\
 7 & 3x3 & 16,422,605  & 13,949         & & & 5	 & 	 1,797,985  \\
 8 & 3x3 &  2,649,261  & 2,050          & & Yes & 8	 & 	 2,649,261 \\
 9 & 3x3 & 16,815,108  & 36,089         & & & 16 & 	 7,787,080  \\
10 & 3x3 &  1,576,778  & 32,213         & & & 13 & 	 8,733,341 \\
11 & 3x3 & 15,148,462  & 43,925         & & Yes & 11 & 	15,148,462 \\
12 & 3x3 & 17,780,452  & 107,900        & & & 7	 & 	16,422,605 \\
13 & 3x3 &  8,733,341  & 43,925         & & & 9	 & 	16,815,108  \\
14 & 3x3 & 18,413,439  & 1,525,688      & & & 12 & 	17,780,452  \\
15 & 3x3 & 18,071,120  & 2,950,149      & & Yes & 15 & 	18,071,120 \\
16 & 3x3 &  7,787,080  & 95,524,079     & & & 14 & 	18,413,439  \\
17 & 3x3 &    152,890  & 374,676,383    & & & 6	 & 	21,115,998 \\ \hline 

18 & 5x2 & 99,152,813  & 4,098          & & & 24 & 	   397,553 \\
19 & 5x2 & 68,312,662  & 4,098          & & & 25 & 	51,991,303 \\
20 & 5x2 & 58,627,384  & 4,097          & & & 23 & 	58,580,865  \\
21 & 5x2 & 58,620,022  & 4,097          & & & 22 & 	58,580,871  \\
22 & 5x2 & 58,580,871  & 4,096          & & & 21 &	58,620,022 \\
23 & 5x2 & 58,580,865  & 4,096          & & & 20 &	58,627,384  \\
24 & 5x2 &    397,553  & 1,915          & & & 19 &	68,312,662 \\
25 & 5x2 & 51,991,303  & 1,471          & & & 26 & 	78,774,278 \\
26 & 5x2 & 78,774,278  & 501            & & & 18 & 	99,152,813 \\ \hline

27 & 2x5 & 67,639,951  & 90,604         & & & 47  &	 1,249,263 \\
28 & 2x5 & 58,764,276  & 64,665         & & & 49  & 	 1,267,093 \\
29 & 2x5 &  1,724,449  & 97,104         & & & 50  & 	 1,267,591 \\
30 & 2x5 & 48,504,073  & 458,357        & & & 48  & 	 1,267,697 \\
31 & 2x5 & 59,334,694  & 668,420        & & & 29  & 	 1,724,449 \\
32 & 2x5 & 44,049,832  & 1,957,771      & & & 41  & 	13,130,674 \\	 
33 & 2x5 & 22,975,390  & 1,137,477      & & & 40  & 	13,143,213 \\	 
34 & 2x5 & 29,026,306  & 2,576,467      & & Yes & 42  & 	15,559,001 \\	 
35 & 2x5 & 43,147,682  & 4,848,239      & & Yes & 45  & 	15,973,730  \\	 
36 & 2x5 & 60,397,442  & 143            & & & 38  & 	20,503,550 \\
37 & 2x5 & 47,902,356  & 4,099          & & Yes & 39  & 	20,685,942 \\
38 & 2x5 & 20,503,550  & 3,685          & & Yes & 33  & 	22,975,390 \\
39 & 2x5 & 20,685,942  & 11,120         & Yes & Yes & 34  & 	29,026,306 \\
40 & 2x5 & 13,143,213  & 36,543,045     & & & 46  & 	31,987,521 \\
41 & 2x5 & 13,130,674  & 114,668,733    & & & 35  & 	43,147,682 \\
42 & 2x5 & 15,559,001  & 398,005,342    & & & 32  & 	44,049,832  \\
43 & 2x5 & 72,578,263  & 620,906,587    & & & 37  & 	47,902,356 \\
44 & 2x5 & 49,459,622  & (10 digits)    & & Yes  & 30  & 	48,504,073 \\
45 & 2x5 & 15,973,730  & (10 digits)    & & & 44  & 	49,459,622  \\
46 & 2x5 & 31,987,521  & (12 digits)    & & & 28  & 	58,764,276 \\
47 & 2x5 &  1,249,263  & (31 digits)    & & & 31  & 	59,334,694 \\
48 & 2x5 &  1,267,697  & (106 digits)   & & & 36  & 	60,397,442 \\
49 & 2x5 &  1,267,093  & (106 digits)   & & & 27 & 	67,639,951 \\
50 & 2x5 &  1,267,591  & (353 digits)   & & & 43  & 	72,578,263 \\
\end{tabular}
\caption{Locations of known dreadful dragons}
\label{table:dragons}
\end{table}

\begin{figure}
\centering
\begin{tabular}{cc}
\includegraphics[width=0.5\textwidth]{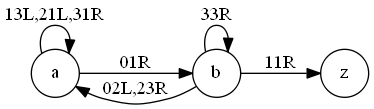} & \includegraphics[width=0.5\textwidth]{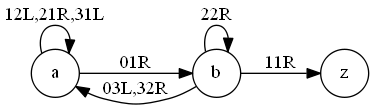} \\
\end{tabular}
\caption{Dragons 2 and 3}
\label{fig:machines2and3}
\end{figure}

\begin{figure}
\centering
\begin{tabular}{cc}
\includegraphics[width=0.5\textwidth]{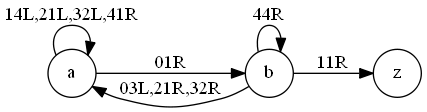} & \includegraphics[width=0.5\textwidth]{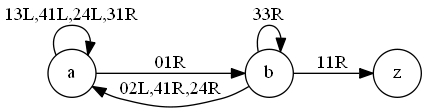} \\
\end{tabular}
\caption{Dragon 39 original and modified}
\label{fig:dragon39}
\end{figure}

\begin{figure}
\centering
\includegraphics[width=0.5\textwidth]{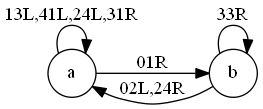} 
\caption{Dragon 39 incomplete definition}
\label{fig:dragon39partial}
\end{figure}

\begin{figure}
\centering
\includegraphics[width=0.5\textwidth]{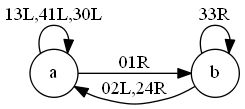} 
\caption{Partially defined machine which does not terminate}
\label{fig:partial4}
\end{figure}

There are a couple of interesting aspects of the generation proces that are highlighted by Table~\ref{table:dragons}. 
Consider first Dragons 2 and 3 (i.e.\ numbers 2 and 3 in Table~\ref{table:dragons}). The definition of these machines is given in Figure~\ref{fig:machines2and3}. These are taken from Marxen's comprehensive list of busy beaver machines \cite{Marxen}, and are labelled as \#b and \#c respectively in the 2x4 class as originally discovered by Pascal Michel. 
It was not difficult to find Dragon 2 in the list of generated $2 \times 4$ machines. However, finding Dragon 3 proved more difficult, and in fact it did not appear to be in the list. On closer inspection, it can be seen that Dragon 3 is in fact productivity and activity equivalent to Dragon 2 (as specified in Lemma~\ref{lemma:symbolname} above), as swapping symbols 2 and 3 in the definition of Dragon 3 gives us Dragon 2. This means that the restrictions we use to generate our machines should also be applied to machines found by other sources in order to do such comparisons. 

This is the reason that the \textbf{Modified} column is included, to show whether or not the definition of the machine as given by Marxen needs to be modified in order to be found in our list of machines. This may involve renaming states and symbols, as well as other transformations given in Section~\ref{sec:normal} which preserve productivity and activity. As can be seen from Table~\ref{table:dragons}, 12 out of the 50 dragons required modification in this way. 

Another interesting aspect of our search is indicated by Dragon 39. The definition given by Marxen, where it is identified as \#c in the $2 \times 5$ class and discovered by T.J. \& S. Ligocki, and our modification of it are given in Figure~\ref{fig:dragon39}. Once this modification was done, we searched for the machine, and found the incomplete definition of it given in Figure~\ref{fig:dragon39partial} generated as number 20,685,942. This has 8 of the same transitions as the modified machine above, but does not include transitions for the $b,4$ and $b,1$ cases. 

The reason that only the incomplete definition was found is due to the bound placed on the amount of execution allowed during the \textit{tnf} search process. The eighth transition added to this machine is $(a,2,4,l,a)$, which is added after 31 execution steps, as this is the first time that the machine is in state $a$ and encounters $2$ as input. This partial definition stays the same until step 1,376, which is the first time that the machine is in state $b$ and encounters $4$ as input. At this point, the $(b,4,1,r,a)$ transition is added, and as there is only one more transition in the machine to be defined, this must be the halting transition, and so the complete machine has been found. However, during the search process, the limit on the number of steps to be executed is set at 200. This means that only the incomplete definition was found (i.e.\ the state of the machine after 200 steps, in which there were 8 transitions known, not 10). 

One reasonable conclusion from this result is that the bound on the number of execution steps needs to be increased, so that this machine (and the other 20 machines with the same first 8 transitions)
are found in the generation process. This will of course increase the time taken to generate the machines, but as this generation is intended to be only done once, this is not the most important consideration. The trickier part is to know what value is appropriate for this bound. For Dragon 39, a bound of say 1,500 would be sufficient. However, it is far from clear whether this bound is appropriate for all machines in the $2 \times 5$ class, or even for a substantial fraction of the other incompletely defined machines generated. It seems fundamental to have some such limit, as there will of course be some incompletely defined machines which do not terminate. For example, the machine in Figure~\ref{fig:partial4} does not terminate, as execution of this machine continually results in configurations of the form $1\{a\}3 2^n 132$ for $n \geq 1$. Accordingly, the \textit{tnf} process will never generate any machine with a strict superset of these seven transitions.

Unless we were to abandon entirely the process of separating the generation of machines from their analysis, it seems that the problem of such incompletely defined machines will remain. We can set the bound on execution during the search process to something that appears reasonable, but there seems to be no way to determine \textit{in advance} a value of this bound which will ensure that the only incompletely defined machines which exceed this bound are those which do not terminate. There must, in fact, be such a bound; in fact, the busy beaver value for the relevant class of machines is an upper bound on this number. However, as the busy beaver value is precisely what we are ultimately trying to determine, this information is not very helpful in setting an appropriate value for this bound. In addition, while the time taken to generate the machines is not of great concern, it cannot be totally ignored, especially given the lack of any guarantee of success. With a bound of 200, the generation of the $2 \times 5$ machines takes around 13 hours; increasing this bound to say 1,500 will clearly significantly increase this time, possibly to around 90 hours, and it is not clear in advance how much difference this extended generation time will make to the number of incompletely defined machines that are generated. 

For now, we note that this issue exists, and that we anticipate the number of such incompletely defined machines is relatively small, based purely on the data that of the 50 known dreadful dragons in the machines up to and including the \textbf{White Whale} (i.e.\ machines 1-50 in \cite{Harland16}), only 1 machine exhibited this problem. Naturally the precise number of such machines (i.e.\ those which are not completely defined by the generation process, and which we cannot show to be non-terminating) will not be known until the class of machines is analysed (i.e. step 3 of the process in \cite{Harland16}). One outcome of such an analysis may well be that there are too many such machines, which means that the bound used in the search needs to be larger, and so the generation must be performed again with this larger bound. If, however, the number of such machines is relatively small (say less than 1 machine in 10,000), then it may be reasonable to keep the same generated class of machines, and perform a more extensive \textit{tnf} search on this relatively small number of machines during the analysis step. 

\section{Conclusions and Further Work}\label{sec:conclusion}

This paper is the second in what may be a lengthy sequence aimed at fleshing out the framework of \cite{Harland16}. In this paper we have concentrated on steps 1 and 2, which involve the generation of the classes of machines which need to be analysed. As we have seen, it is not altogether straightforward to separate the generation of the machines from their analysis, but doing so as best we can seems appropriate. This will not only allow for separate and independent analyses of the machines, but also for the incremental development of analysis techniques. The latter may seem to be a rather trite purpose, but given the apparent complexity of the analysis of some machines, it seems altogether prudent. 

We have given a precise specification of the busy beaver problem, both in terms of the specific variety of Turing machines that we use, as well as the central issues around the productivity and activity of these machines. We have seen how quintuple machines are at least as general as quadruple machines, which means that searching among the quintuple machines will not omit any machines in the quadruple ones, and possibly allow for more cases than the quadruple machines allow. We have given formal results which establish the soundness of our constraints on the machines to be generated, and given a procedure to generate machines satisfying these constraints. We have implemented this procedure and reported our results, which together with all the code used, is available on the author's website. It is hope that this code and data will be useful to other researchers interested in this problem. 

There remain some unresolved issues arising from the results of this paper. One such item of further work is to provide a proof of Conjecture~\ref{conj:symbols}, i.e.\ that the $bb$ and $\mathit{ff}$
 functions are strictly monotonic in the number of symbols used in the machine. This seems obviously true, but as discussed in Section~\ref{sec:generation}, the difficulty is finding a way to extend a terminating $m$-symbol machine into a terminating $(m+1)$-symbol machine of strictly greater activity and productivity. It may be that productivity preserving transformations along the lines of those of Shannon \cite{Shannon56} will be useful here.

Another issue is to investigate whether a similar strict monotonicity result holds for the number of transitions in a machine. Chaitin \cite{Chaitin87} has argued that this is a more natural way in which to organise the busy beaver problem, in that this more directly reflects the complexity of the computation specified by the machine than the number of states or symbols used (or both). 
This suggests that it may be interesting to re-stratify the classes of machines up to and including the \textbf{White Whale}, once the analysis of the machines is done. Not only would this potentially remove some redundancies, it may also introduce some finer-grained analysis, such as comparing the maximum productivity of a 5-state 2-symbol machine with 8 transitions with a similar maximum for machines with 9 or 10 transitions. Such a reworking of the \textbf{White Whale} results may also yield some useful insights for tackling the \textbf{Demon Duck of Doom}. 

A further item of future work is to sharpen our knowledge of the relationship between the quadruple and quintuple variants. As we have seen in Section~\ref{sec:qq}, the quintuple machines are at least as general as the quadruple ones, which is a justification of our choice to concentrate on quintuple machines. It seems that a stronger result is possible, i.e.\ that there are quintuple machines for which there are no productivity-equivalent quadruple machines. It is certainly possible to show that certain ways of simulating quintuple machines on quadruple ones will not work, and even if there is one specific quintuple machine that can be shown inequivalent to any quadruple machine of the same dimension, that would suffice to establish the required separation. However, the difficulty is being able to establish the inequivalence of such a given quintuple machine with \textit{any} quadruple machines of the same dimension. It may be more effective to perform an enumeration of the quadruple machines similar to the one described in this paper in order to understand the relationships between these variants in more detail. 

The next step in the process of \cite{Harland16} is to analyse the generated machines, and in particular to classify them as either terminating or non-terminating. This is likely to take significantly more effort than the results reported in this paper. We also expect that this analysis will suggest further items of theoretical interest, just as this phase of the process has done. 

\section*{Acknowledgements}

The author would like to thank
Jeanette Holkner, 
Barry Jay, 
Pascal Michel,
Sandra Uitenbogerd, 
Michael Winikoff,
and some anonymous referees for valuable feedback on this material. 

\bibliographystyle{plain}
\bibliography{busybeaver}

\end{document}

%% file: bbdefs.tex
\usepackage{graphicx}
 \usepackage{amsfonts}
 \usepackage{amssymb}
 \usepackage{fullpage}
  \usepackage{url}
\newenvironment{proof}%
        {\begin{list}{ {\bf Proof}}%
                           {\setlength{\leftmargin}{\labelwidth}} \item }%
        {\end{list}}

\newtheorem{definition}{Definition}

\newtheorem{lem}[definition]{Lemma}
\newtheorem{prop}[definition]{Proposition}

\newtheorem{conj}[definition]{Conjecture}